\documentclass[12pt,journal,draftclsnofoot,onecolumn,letterpaper]{IEEEtran}

\usepackage{mathrsfs}
\usepackage{amssymb}
\usepackage{amsmath}
\usepackage{amsmath,bm}
\usepackage{amsthm}
\usepackage{graphicx}
\usepackage{subfigure}
\usepackage{cite}
\usepackage{enumerate}
\usepackage{color, soul}
\usepackage{verbatim}
\usepackage{amsfonts}
\usepackage{array}
\usepackage{bm}
\usepackage{subfigure}
\usepackage{diagbox}
\usepackage{algorithmic}
\usepackage{algorithm}

\begin{document}

\newtheorem{definition}{\bf ~~Definition}
\newtheorem{observation}{\bf ~~Observation}
\newtheorem{lemma}{\bf ~~Lemma}
\newtheorem{proposition}{\bf ~~Proposition}
\newtheorem{remark}{\bf ~~Remark}
\renewcommand{\algorithmicrequire}{ \textbf{Input:}} %Use Input in the format of Algorithm
\renewcommand{\algorithmicensure}{ \textbf{Output:}} %UseOutput in the format of Algorithm
\newcommand*{\TitleFont}{%
      %\usefont{\encodingdefault}{\rmdefault}{b}{n}%
%      \fontsize{20}{20}%
      \selectfont}
\title{\TitleFont Cellular UAV-to-Device Communications:  Trajectory Design and Mode Selection by Multi-agent Deep Reinforcement Learning}

\author{
\IEEEauthorblockN{
{Fanyi Wu}, \IEEEmembership{Student Member, IEEE},
{Hongliang Zhang}, \IEEEmembership{Member, IEEE},\\
{Jianjun Wu},
{and Lingyang Song}, \IEEEmembership{Fellow, IEEE}}\\

\vspace{-0.5cm}
\thanks{Part of this paper was presented at IEEE GLOBECOM, Waikoloa, HI, USA, Dec.~2019~\cite{FHJL-2019}.}
\thanks{F.~Wu, J.~Wu, and L.~Song are with Department of Electronics Engineering, Peking University, Beijing 100871, China~(email: fanyi.wu@pku.edu.cn, just@pku.edu.cn, lingyang.song@pku.edu.cn).}
\thanks{H.~Zhang is with Department of Electronics Engineering, Peking University, Beijing 100871, China, and also with Department of Electrical and Computer Engineering, University of Houston, Houston, TX 77004, USA (email: hongliang.zhang92@gmail.com).}

}

\maketitle

\vspace{-0.5cm}

\begin{abstract}
In the current unmanned aircraft systems~(UASs) for sensing services, unmanned aerial vehicles~(UAVs) transmit their sensory data to terrestrial mobile devices over the unlicensed spectrum. However, the interference from surrounding terminals is uncontrollable due to the opportunistic channel access. In this paper, we consider a cellular Internet of UAVs to guarantee the Quality-of-Service~(QoS), where the sensory data can be transmitted to the mobile devices either by UAV-to-Device (U2D) communications over cellular networks, or directly through the base station~(BS). Since UAVs' sensing and transmission may influence their trajectories, we study the trajectory design problem for UAVs in consideration of their sensing and transmission. This is a Markov decision problem~(MDP) with a large state-action space, and thus, we utilize multi-agent deep reinforcement learning~(DRL) to approximate the state-action space, and then propose a multi-UAV trajectory design algorithm to solve this problem. Simulation results show that our proposed algorithm can achieve a higher total utility than policy gradient algorithm and single-agent algorithm.
\end{abstract}


\newpage

%%%%%%%%%%%%%%%%%%%%%%%%%%%%%%%%%%%%%%%%%%%%%%%
\section{Introduction}%
%%%%%%%%%%%%%%%%%%%%%%%%%%%%%%%%%%%%%%%%%%%%%%%
\label{Introduction}

With high mobility and low operational cost, unmanned aerial vehicle~(UAV) is recognized as a powerful facility to provide sensing services~\cite{HJJYG-2018,JCZYRL-2017}, which has found use in a wide range of sensing applications including traffic monitoring~\cite{KGMK-2013}, industrial inspection~\cite{JMJSCR-2013}, precision agriculture~\cite{BKDF-2017}, and fire surveillance~\cite{CZY-2017}. In the current unmanned aircraft systems~(UASs) for sensing services, UAVs transmit their sensory data to terrestrial mobile devices\footnote{These mobile devices can be smartphones, laptops, and dedicated ground-based controllers. Besides, the sensory data are collected by the payloads of UAVs (e.g., cameras or thermometers), which can be further processed on the mobile devices.} over the unlicensed spectrum~\cite{RJW-2013}. However, due to the opportunistic channel access at the media access control~(MAC) layer, the interference from surrounding terminals is uncontrollable, and the Quality-of-Service~(QoS) for sensing services cannot be guaranteed. Therefore, a more reliable network is necessary.

As an effective solution, terrestrial cellular networks can be utilized to provide UAV sensing services with guaranteed QoS, which we refer to as the cellular Internet of UAVs~\cite{3GPP_TR_36_777,HLZH-2019}. In the cellular Internet of UAVs, UAVs first transmit the sensory data to the BS, after which the BS sends the received data to mobile devices. When UAVs fly far from their sensing targets, they may suffer low sensing quality, while if UAVs are far from the BS, it is difficult to transmit their sensory data back to the BS. Therefore, to ensure the QoS, UAVs should jointly take sensing and transmission into consideration in the designing of their trajectories. However, when UAVs are close to their mobile devices, the data rate can be further improved if UAVs can directly transmit the sensory data to proximal mobile devices, namely the UAV-to-Device~(U2D) communications, rather than transmit through cellular communications.

In this paper, we propose to enable U2D communications in the cellular Internet of UAVs. We consider a orthogonal frequency division multiple access~(OFDMA) cellular Internet of UAVs\footnote{A typical application of this network is traffic monitoring. In the application, UAVs are required to record videos from their target roads to collect the road condition information. To realize real-time monitoring, the recorded videos should be transmitted to their corresponding terrestrial control stations simultaneously, which operates over the cellular spectrum to guarantee the QoS.}, where the sensory data can be transmitted to the mobile devices in either the U2D or the cellular mode, and the U2D communications work as an overlay to the cellular network. Since the trajectories of UAVs may influence both sensing and transmission, it is important to design their trajectories. Moreover, UAVs' trajectories are also coupled with their sensing and transmission, and thus, it is challenging to solve the UAV trajectory design problem.

To tackle with this challenge, we first design a joint sensing and transmission protocol to enable the joint sensing and transmission for UAVs in different modes. This protocol can be analyzed by nested bi-level Markov chains~\cite{PX-2017}, in which sensing and transmission processes are formulated as the state transitions in the Markov chains. Therefore, the UAV trajectory design problem can be regarded as a Markov decision problem~(MDP)~\cite{RA-1998,JHLZH-2019}. Since the state-action space is very large, we then analyze this problem using multi-agent deep reinforcement learning~(DRL)~\cite{KMMA-2017,Y-2017}, and propose a multi-UAV trajectory design algorithm based on deep Q-networks~(DQNs) to obtain the optimal policies for UAVs.

In the literature, several works have investigated the cellular Internet of UAVs, which focus on the UAV-to-Network (U2N)~\cite{SHBL2-2018} communications that UAVs could set up links with the BS, and UAV-to-UAV~(U2U)~\cite{SHBL3-2018} communications that multiple UAVs could communicate with each other directly. Specifically, authors in~\cite{SHBL2-2018} investigated a cellular network consisting of one UAV for sensing services, and maximized the energy efficiency in the network by jointly optimizing the UAV's trajectories as well as transmission power. In~\cite{SHBL3-2018}, the authors jointly optimized subchannel allocation and flying speed for a cellular Internet of UAVs, where UAVs could build up U2N links with the BS as well as U2U links with other UAVs. Nevertheless, as an important practical scenario in the cellular Internet of UAVs,
the direct links among UAVs and mobile devices, i.e., U2D communications, are lack of investigation in the current works\footnote{Actually, there still exist several works on the transmission between UAVs and the mobile devices~\cite{SHQKL-2018,MMWCMC-2017}. However, in these works, UAVs are adopted as relays~\cite{SHQKL-2018} or BSs~\cite{MMWCMC-2017}, which are different from our work where UAVs serve as Internet-of-Things (IoT) devices to execute sensing tasks~\cite{ZYNDPN-2020}. In our paper, UAV sensing and transmission are jointly investigated, which is more complicated than the UAV networks without the consideration of UAV sensing.} .

Besides, there also exist works focusing on the configuration of the UAV networks where UAVs are adopted as the moving BSs~(UAV-BSs)~\cite{ISH-2018,FIH-2018,IAH-2019}. Authors in~\cite{ISH-2018} studied the positioning of UAV-BSs in the urban areas, and proposed an advanced channel model in consideration of the urban environment. In~\cite{FIH-2018}, the authors investigaed the strategic deployment of UAV-BSs in the presence of a terrestrial network, in which the amount and the locations of UAV-BSs are investigated.  A spatial network configuration scheme was proposed in~\cite{IAH-2019} to exploit the mobility for an agile wireless network with UAV-BSs. Different from the above works, in our paper, UAVs are utilized to provide sensing services for terrestrial mobile devices, in which multiple UAVs' sensing and transmission should be jointly analyzed.

Furthermore, several works have studied the machine learning~(ML)~\cite{UWC-2019,MWC-2019} in UAV networks. In~\cite{UWC-2019}, the authors jointly investigated energy efficiency, wireless latency, and interference management for a UAV networks leveraging dynamic game, and then designed a deep reinforcement learning algorithm based on echo state network to achieve the subgame perfect Nash equilibrium. Authors in~\cite{MWC-2019} jointly investigated cache management and resource allocation for a UAV network over both licensed and unlicensed bands, and then proposed a distributed ML algorithm based on the liquid state machine to solve this problem. Unlike the above works in which UAVs are utilized as flying infrastructures, e.g., base station or relay, which are different from our work where UAVs serve as aerial Internet-of-Things (IoT) devices to execute sensing
tasks. In this system, UAV sensing and transmission should be jointly investigated, and thus, the ML methods in the current works cannot be adopted directly in our problem.

The main contributions of this paper can be summarized as follows:
\begin{itemize}
  \item We propose overlaying U2D communications in the cellular Internet of UAVs to improve the QoS for UAV sensing services, and design a joint sensing and transmission protocol to enable U2D communications. %Then, we design a joint sensing and transmission protocol to enable the joint sensing and transmission, and analyze the protocol using nested bi-level Markov chains.
  \item We formulate the UAV trajectory design problem using multi-agent DRL, and then propose a DQN-based multi-UAV trajectory design algorithm to solve this problem.
  \item Simulation results show that our proposed algorithm can achieve a higher utility in the network than policy gradient algorithm and single-agent algorithm.
\end{itemize}

The rest of the paper is organized as follows. In Section~\ref{System Model}, we describe the model of
the U2D communications overlaying cellular network. In Section~\ref{Sense-and-Send Protocol}, we design a joint sensing and transmission protocol to coordinate multiple UAVs performing sensing tasks. Then, we analyze the protocol using nested bi-level Markov chains in Section~\ref{Markov Analysis}. We investigate on the UAV trajectory design problem using multi-agent DRL, and propose a DQN-based multi-UAV trajectory design algorithm in Section~\ref{Trajectory Design using DRL}. In Section~\ref{System Performance Analysis}, we analyze the convergence and the complexity of our proposed algorithm, and present some key properties on mode selection and trajectory design for UAVs. Simulation results are presented in Section~\ref{Simulation Results}. Finally, Section~\ref{Conclusion} concludes the paper.

%%%%%%%%%%%%%%%%%%%%%%%%%%%%%%%%%%%%%%%%%%%%%%%
\section{System Model}%
%%%%%%%%%%%%%%%%%%%%%%%%%%%%%%%%%%%%%%%%%%%%%%%
\label{System Model}
In this section, we first describe the model for a U2D communications overlaying cellular network. Then, the sensing and the transmission models for UAVs are elaborated on, respectively. For clarity, we summarize all the following notations and their definitions in Table~\ref{Notations}.

\begin{table}[!t]
\centering
\scriptsize
\caption{List of Notations} \label{Notations}
\vspace{-3mm}
\begin{tabular}{|c|p{1.65in}|c|p{2.15in}|}
 \hline
 \textbf{Notations} & \textbf{Definitions} & \textbf{Notations} & \textbf{Definitions} \\
 \hline
 \hline
  $N$ & Number of UAVs  & $\mathcal{T}_i$ & Trajectory \\
  \hline
  $K$ & Number of subchannels  & $\Delta$ & Distance between two adjacent points \\
  \hline
  $P^u$ & Transmit power of UAVs  &$v_{max}$ & Maximum flying speed of UAVs \\
  \hline
  $P^b$ & Transmit power of the BS  & $\mathcal {U}_i$ & Utility \\
  \hline
  $N_0$ & Noise power  &  $r_i$ & Reward \\
  \hline
  $f_c$ & Carrier frequency  & $\pi_i$ & Policy of agent \\
  \hline
  $H_0$ & Height of the BS & $\mathcal {S} = \{ \emph{\textbf{s}} \} $ & State space \\
  \hline
  $l_i(\textbf{\emph{x}}_1,\textbf{\emph{x}}_2)$ & Distance between $\textbf{\emph{x}}_1$ and $\textbf{\emph{x}}_2$ & $\mathcal {A}_i = \{ \emph{\textbf{a}}_i \}$  & Action space \\
  \hline
  $\lambda$ & UAV sensing factor & $\mathcal {P}$ & State transition function \\
  \hline
  $t_f$ & Duration of a frame & $\mathcal {R}_i(\emph{\textbf{s}}, \emph{\textbf{a}}_i )$ & Reward function \\
  \hline
  $\mathcal{P}_{ss,i}$ & Successful sensing probability & $\rho$ &Discount factor \\
  \hline
  $\gamma_{a,i}$, $\gamma_{g,i}$ & Signal-to-noise ratio in air-to-ground/ground-to-ground link & $p^u_{ts,i}$, $p^c_{ts,i}$ & The probability that transmission succeeds in U2D/cellular mode within a cycle \\
 \hline
 $\mathcal {L}_{a,i}$, $\mathcal {L}_{g,i}$ & Path loss in air-to-ground/ground-to-ground link  & $\mathbb{P}^u_{ts,i}(t)$, $\mathbb{P}^c_{ts,i}(t)$ & The probability that transmission succeeds in U2D/cellular mode within the $t$-th frame \\
  \hline
$\zeta_i$, $\kappa_i$ & Small-scale fading coefficient in air-to-ground/ground-to-ground link & $\mathbb{P}^u_{ts,i}\{t|\emph{\textbf{I}}(t)\}$, $\mathbb{P}^c_{ts,i}\{t|\emph{\textbf{I}}(t)\}$ & The probability that transmission succeeds in U2D/cellular mode after the $t$-th frame, given the state indicator at the $t$-th frame $\emph{\textbf{I}}(t)$ \\
  \hline
 $R^u_i$, $R^c_i$ & Throughput in U2D/cellular mode & $Q_i(\emph{\textbf{s}}, \emph{\textbf{a}}_i )$ & Q-value function \\
 \hline
  $R_{th}$ & QoS requirement  &  $Q_i$,  $\Theta_i$ & Q network and its weight \\
 \hline
 $\mathcal{P}^{u}_{ts,i}$, $\mathcal{P}^{c}_{ts,i} $ & Successful transmission probability in U2D/cellular mode & $\widehat{Q}_i$, $\Theta_i^-$ & Target network and its weight \\
 \hline
 $T_s$, $T_t$, $T_c$ & Length of sensing part/transmission part/full cycle & $\mathcal {M}$, $M$ & Replay memory and its size \\
 \hline
 $\emph{\textbf{I}}(t)$ & Transmission state indicator & $\mathcal {D}$, $D$& Mini-batch and its size \\
 \hline
 $\textbf{\emph{v}}(t)$ & Subchannel allocation vector & $\tau$ & Update frequency of target networks \\
 \hline
\end{tabular}
\vspace{-7mm}
\end{table}

\subsection{System Description}

\begin{figure}[!t]
\centering
\vspace{-5mm}
\includegraphics[width=6in]{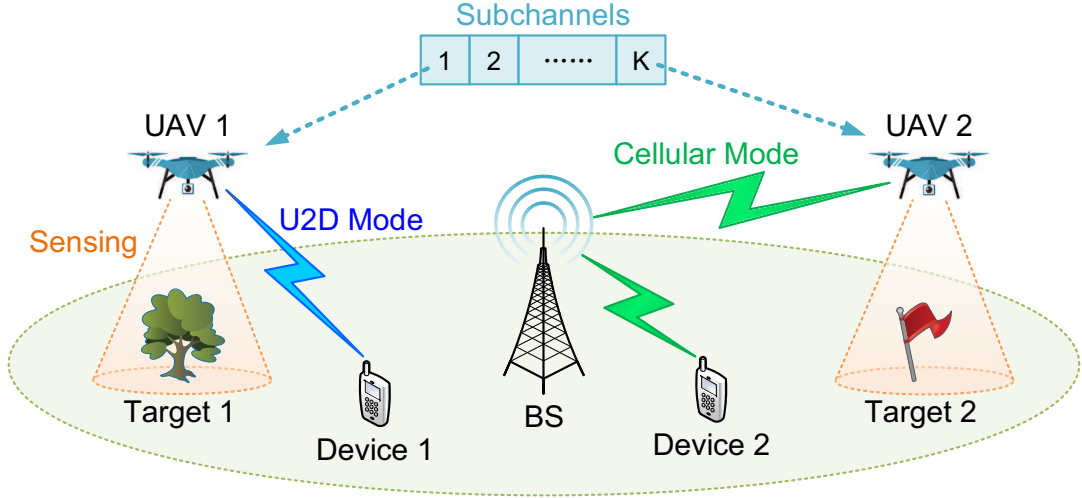}
\vspace{-9mm}
\caption{System model for the U2D communications overlaying cellular network.}
\vspace{-7mm}
\label{system_model}
\end{figure}

As illustrated in Fig.~\ref{system_model}, we consider an OFDMA cellular network with single BS. The system consists of $N$ UAVs performing sensing tasks, which denoted by $\mathcal {N} = \{1,2,...,N\}$. Each UAV is required to sense its target, and then transmit the sensory data to the corresponding mobile device on the ground. There exist two modes to support the data transmissions from a UAV to its mobile device:
\begin{itemize}
  \item \textbf{U2D Mode}: The UAV transmits the sensory data to its corresponding mobile device directly;
  \item \textbf{Cellular Mode}: The data transmissions include two phases. In the first phase, the UAV sends the sensory data to the BS. In the second phase, the BS transmits the received data to the mobile device.
\end{itemize}

We define the time unit for the transmission as frame. To be specific, a UAV in the U2D mode takes a whole frame to send the sensory data to its mobile device directly. However, the two phases in the cellular mode are performed in a time division multiplexing~(TDM) manner. The UAV in the cellular mode sends the sensory data to the BS in the first half frame, and then the BS transmits the received data to the mobile device in the other half frame.

The network owns $K$ orthogonal subchannels, denoted by $\mathcal {K} = \{1,2,...,K\}$, to support the OFDMA transmissions among UAVs and mobile devices. In order to avoid the mutual interference among UAVs, the BS will allocate orthogonal subchannels to different UAVs\footnote{For UAVs in the cellular mode, the two phases are operated over the same subchannel.}. Besides, we assume that the available subchannels are not sufficient to support the data transmissions of all UAVs simultaneously, and thus, we have $K \leq N$.

In the system, the locations of the BS, UAVs, mobile devices and sensing targets are specified by 3D cartesian coordinates. Specifically, the location of the BS is denoted as $\textbf{\emph{x}}_0 = (0, 0, H_0)$ with $H_0$ being its height. The location of the $i$-th UAV is denoted as $\textbf{\emph{x}}_i = (x_i, y_i, h_i)$, and its corresponding mobile device locates at $\textbf{\emph{x}}^d_i = (x^d_i, y^d_i, 0)$. Moreover, the coordinates of the sensing target for $i$-th UAV can be given by $\textbf{\emph{x}}^t_i = (x^t_i, y^t_i, 0)$.

%In the following, we will introduce the sensing and transmission models for UAVs, respectively.

\subsection{UAV Sensing}
We utilize the probabilistic sensing model in \cite{VI-2017,ARAS-2013} to evaluate the sensing qualities of UAVs. Specifically, the successful sensing probability~(SSP) for a UAV is an exponential function of the distance between the UAV and its target. When the $i$-th UAV senses its target for a frame, the SSP can be given as
\begin{equation}\label{SSP}
  \mathcal{P}_{ss,i} = e^{-\lambda t_f l_i},
\end{equation}
where $\lambda$ is the parameter evaluating the sensing performance, $t_f$ is the duration of a frame, and $l_i = \|\textbf{\emph{x}}_i - \textbf{\emph{x}}^t_i\|_2$ denotes the distance from the $i$-th UAV to its target.

Due to the limited computational capability, each UAV cannot figure out whether the sensing is successful or not by itself. Instead, mobile devices can judge whether the sensory data received from UAVs is valid or not. However, we can still evaluate the sensing qualities of UAVs by calculating their SSPs based on equation~(\ref{SSP}).

\subsection{UAV Transmission}
%%%%%%%%%%%%%%%%%%%%%%%%%%%%%%%%%%%
\label{Transmission Model}
%%%%%%%%%%%%%%%%%%%%%%%%%%%%%%%%%%%

Since UAVs fly at high altitudes, the line-of-sight~(LoS) components usually exist in the transmissions from UAVs to their mobile devices and the BS on the ground~\cite{SJHL-2019,HLZ-2020}. Thus, the channel characteristic of the air-to-ground communications is different from that in traditional terrestrial communications. For UAVs in the U2D mode, the sensory data is sent on an air-to-ground channel. However, when UAVs select the cellular mode, the sensory data is transmitted on an air-to-ground channel in the first phase, and then on a terrestrial channel in the second phase. In what follows, we will analyze the transmissions in the U2D and the cellular modes, respectively.

\subsubsection{U2D Mode}
We adopt the air-to-ground channel model in \cite{3GPP_TR_36_777} to evaluate the direct data transmissions from UAVs to their mobile devices. Then, the received signal-to-noise ratio~(SNR) from the $i$-th UAV to its mobile devices can be expressed as
\begin{equation}\label{U2D_SNR}
\gamma_{a,i}(\textbf{\emph{x}}^d_i,\textbf{\emph{x}}_i) = \frac{P^u \zeta_i}{N_0 \cdot 10^{\mathcal {L}_{a,i}/10}}
\end{equation}
where $P^u$ is the transmit power of UAVs, $N_0$ denotes the power of noise, $\mathcal {L}_{a,i}$ denotes the air-to-ground path loss, and $\zeta_i$ is the small-scale fading coefficient.

To calculate $\mathcal {L}_{a,i}$ and $\zeta_i$, we consider both the LoS and the none LoS~(NLoS) components. According to~\cite{3GPP_TR_36_777}, the LoS probability for the $i$-th UAV can be calculated by
\begin{equation}
\mathcal{P}_{LoS,i}=
\left\{
\begin{aligned}
1,\qquad\qquad\qquad\quad&\quad{d_{2D,i} \leq d_c},\\
{d_{c}}/{d_{2D,i}} + \left( 1 - d_{c}/d_{2D,i}\right)
e^{-d_{2D,i}/{p_0}},&\quad{d_{2D,i}> d_c},\\
\end{aligned}
\right.
\end{equation}
where $d_c = \textmd{max}\{ 294.05\textmd{lg}(h_i)-432.94,18\}$, $d_{2D,i}=\sqrt{(x^d_i-x_i)^2 + (y^d_i-y_i)^2}$, and $p_0 = 233.98 \textmd{lg}(h_i) - 0.95$. Then, the NLoS probability for the $i$-th UAV is $\mathcal{P}_{NLoS,i} = 1 - \mathcal{P}_{LoS,i}$. Based on~\cite{3GPP_TR_36_777}, we can calculate the LoS path loss and the NLoS path loss for the $i$-th UAV, denoted by $\mathcal {L}_{LoS,i}$ and $\mathcal {L}_{NLoS,i}$, respectively. In addition, the LoS small-scale fading coefficient for the $i$-th UAV $\zeta_{LoS,i}$ obeys Rice distribution, and the NLoS small-scale fading coefficient for the $i$-th UAV $\zeta_{NLoS,i}$ obeys Rayleigh distribution. More details on the calculations of the path loss and the small-scale fading coefficient are included in~\cite{3GPP_TR_36_777}.

Based on the received SNR, we can calculate the throughput for the $i$-th UAV in the U2D mode as $R^{u}_{i} = \textmd{log}_2(1+\gamma_{a,i}(\textbf{\emph{x}}^d_i,\textbf{\emph{x}}_i))$. Considering the QoS of transmissions, we define that a transmission is \emph{successful} when the throughput exceed a threshold, denoted by $R_{th}$. Then, the successful transmission probability~(STP) for the $i$-th UAV in the U2D mode can be calculated by $ \mathcal{P}^{u}_{ts,i} = \mathcal{P}_{a,i}(\textbf{\emph{x}}^d_i,\textbf{\emph{x}}_i,R_{th})$, where $\mathcal{P}_{a,i}(\cdot)$ is given by
\begin{equation}\label{U2D_STP}
\mathcal{P}_{a,i}(\textbf{\emph{x}}^d_i,\textbf{\emph{x}}_i,\beta) = \mathcal{P}_{LoS,i} \left( 1-F_{ri}(\chi_{LoS,i}) \right) + \mathcal{P}_{NLoS,i} \left( 1-F_{ra}(\chi_{NLoS,i}) \right),
\end{equation}
in which $\chi_{LoS,i} = ({ (2^{\beta}-1) N_0 \cdot 10^{0.1\mathcal {L}_{LoS,i}}})/{P^u}$, and $\chi_{NLoS,i} = ({(2^{\beta}-1) N_0 \cdot 10^{0.1\mathcal {L}_{NLoS,i}}})/{P^u}$. Here, $F_{ri}(x) = 1 - Q_1(\sqrt{2K_{ri}},x\sqrt{2(K_{ri}+1)})$ is the cumulative distribution function~(CDF) of the Rice distribution with $\Omega = 1$~\cite{S-1944}, $F_{ra}(x) = 1 - e^{-x^2/2}$ is the CDF of the Rayleigh distribution with unit variance, and $Q_1(x)$ is the Marcum Q-function of order 1~\cite{J-1950}.

\subsubsection{Cellular Mode}
In the first phase, the transmissions from the $i$-th UAV to the BS can also be evaluated by the air-to-ground channel model. Therefore, we can express the received SNR at the BS as $\gamma_{a,i}(\textbf{\emph{x}}_0,\textbf{\emph{x}}_i)$, where $\gamma_{a,i}(\cdot)$ is given in~(\ref{U2D_SNR}). As for the second phase, we utilize the free space propagation path loss model with Rayleigh fading~\cite{WCPP} to characterize the terrestrial transmissions. Thus, the received SNR at the mobile device can be expressed as
\begin{equation}\label{cellular_SNR2}
\gamma_{g,i}(\textbf{\emph{x}}^d_i,\textbf{\emph{x}}_0) = \frac{P^b \kappa_i}{N_0 \cdot 10^{\mathcal {L}_{g,i}/10}}
\end{equation}
where $P^b$ is the transmit power of the BS, $\mathcal {L}_{g,i}[\textmd{dB}] = 20 \textmd{lg}(d_i) + 20 \textmd{lg}(f_c) + 32.45$ is the path loss from the BS to the $i$-th mobile device, $d_i[\textmd{m}]= \|\textbf{\emph{x}}^d_i - \textbf{\emph{x}}_0\|_2$ denotes the distance between the BS and the mobile device, $f_c[\textmd{GHz}]$ is the carrier frequency, and $\kappa_i$ denotes the small-scale fading following Rayleigh distribution with unit variance.

The throughput of $i$-th UAV in the cellular mode is given by $R^c_{i} = \frac{1}{2} \min\left\{ R^c_{a,i}, R^c_{g,i} \right\}$, in which $R^c_{a,i} = \textmd{log}_2(1+\gamma_{a,i}(\textbf{\emph{x}}_0,\textbf{\emph{x}}_i))$ and $R^c_{g,i} = \textmd{log}_2(1+\gamma_{g,i}(\textbf{\emph{x}}^d_i, \textbf{\emph{x}}_0))$ denote throughput of the transmissions in the first and the second phases. To achieve successful transmissions, $R^c_i$ should exceed the QoS requirement $R_{th}$, which is equivalent to that both $R^c_{a,i}$ and $R^c_{g,i}$ should be larger than $2R_{th}$. Therefore, we can calculate the STPs for the first and the second phases by $\mathcal{P}_{a,i}(\textbf{\emph{x}}_0, \textbf{\emph{x}}_i, 2R_{th})$ and $\mathcal{P}_{g,i}(\textbf{\emph{x}}^d_i,\textbf{\emph{x}}_0,2R_{th})$, respectively. Here, $\mathcal{P}_{a,i}(\cdot)$ is given in~(\ref{U2D_STP}), and $\mathcal{P}_{g,i}(\cdot)$ is expressed by
\begin{equation}\label{STP_ground}
  \mathcal{P}_{g,i}(\textbf{\emph{x}}^d_i,\textbf{\emph{x}}_0,\beta) = 1-F_{ra}(\chi_{g,i}),
\end{equation}
where $\chi_{g,i} = ((2^{\beta}-1) N_0 \cdot 10^{0.1\mathcal {L}_{g,i}})/{P^b}$, and $F_{ra}(x) = 1 - e^{-x^2/2}$. Since the channel state of these two phases are independent with each other, we can calculate the STP for the $i$-th UAV in the cellular mode by $\mathcal{P}^{c}_{ts,i} = \mathcal{P}_{a,i}(\textbf{\emph{x}}_0, \textbf{\emph{x}}_i, 2R_{th})\mathcal{P}_{g,i}(\textbf{\emph{x}}^d_i,\textbf{\emph{x}}_0,2R_{th})$.

\section{Joint Sensing and Transmission Protocol}
%%%%%%%%%%%%%%%%%%%%%%%%%%%%%%%%%%%%%%%%%%%%%%%%%%%%%%%%
\label{Sense-and-Send Protocol}

In this section, we design a joint sensing and transmission protocol to coordinate multiple UAVs performing sensing tasks simultaneously. We first show the joint sensing and transmission cycles which consists of sensing parts as well as  transmission parts. Then, the subchannel allocation mechanism is presented.

\subsection{Joint Sensing and Transmission Cycle}

\begin{figure}[!t]
\centering
\vspace{-5mm}
\includegraphics[width=6in]{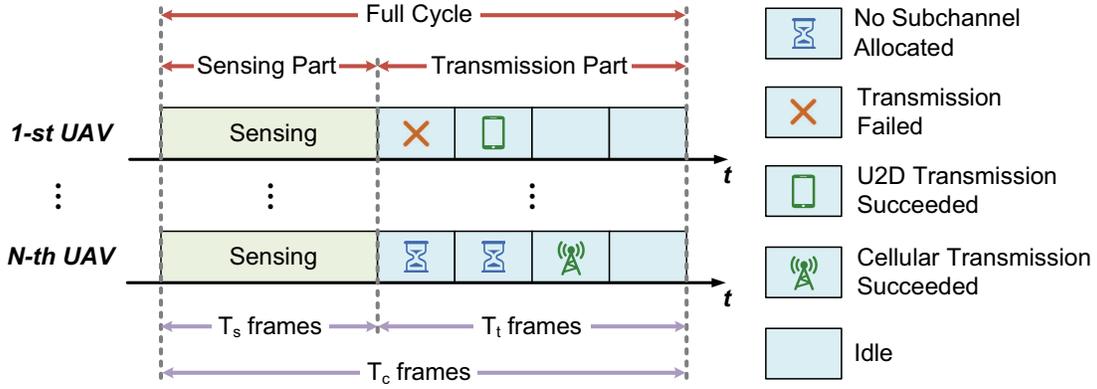}
\vspace{-8mm}
\caption{The joint sensing and transmission protocol.}
\vspace{-7mm}
\label{protocol}
\end{figure}

We assume that the UAVs perform sensing tasks in a synchronized iterative manner, which is characterized by \emph{cycles}. To be specific, in each cycle, UAVs first sense their sensing targets and then transmit the sensory data to their mobile devices. In Fig.~\ref{protocol}, we illustrates a cycle of the joint sensing and transmission protocol, which contains $T_c$ frames. Each cycle is divided into the sensing part with $T_s$ frames and the transmission part with $T_t$ frames. In the sensing part, each UAV should sense its target and collect the sensory data continuously for $T_s$ frames to ensure the sensing quality.

In the transmission part, each UAV attempts to send the sensory data to its mobile device if the BS allocates a subchannel to it. Specifically, at the beginning of each frame, all UAVs should report their locations to the BS, and then the BS preforms the subchannel allocation based on the received information. If a UAV is not allocated to a subchannel, the sensory data cannot be transmitted in this frame. On the other hand, the UAV with allocated subchannel attempts to transmit the sensory data to the mobile device in the transmission mode which satisfies the QoS requirement. When both the U2D and the cellular modes can support successful transmissions, the UAV will select the mode with a higher throughput. However, when neither one of these two modes can satisfy the QoS requirement, the transmission will fail in this frame.

To be specific, there are five possible situations for each UAV in the transmission part:
\begin{itemize}
  \item \textbf{No Subchannel Allocated}: No subchannel is allocated to the UAV, and thus, the UAV cannot transmit data in this frame. It will wait for the BS to assign a subchannel to it in the following frames;
  \item \textbf{Transmission Failed}: A subchannel is allocated to the UAV. However, the transmission
        is unsuccessful due to the QoS requirement. Thus, the UAV will attempt to transmit the sensory data again in the following frames;
  \item \textbf{U2D Transmission Succeeded}: A subchannel is allocated to the UAV, and the UAV successfully transmits the sensory data in the U2D mode;
  \item \textbf{Cellular Transmission Succeeded}: A subchannel is allocated to the UAV, and the UAV successfully transmits the sensory data in the cellular mode;
  \item \textbf{Idle}: After a successful transmission, the UAV will keep idle until the next cycle begins.
\end{itemize}

Finally, at the end of the transmission part, the BS broadcasts to inform UAVs of the locations of all the UAVs. By this means, each UAV can obtain the locations of other UAVs, and then decide its trajectory in the next cycle. In practice, the information exchange between UAVs and the BS can be operated through the control channels~\cite{RJW-2013}, which is separated from the sensory data transmission over the data channels\footnote{The overhead for information exchange is trivial, which can be covered by the control channels owned by the system. We assume that a UAV needs $\omega$ messages to describe its location. Thus, the overhead for information exchange is $N \omega$ messages, in which $N$ is the number of UAVs.}.

\subsection{Subchannel Allocation Mechanism}
\label{Subchannel Allocation Mechanism}

In the transmission part, the available subchannels may not be sufficient to support all the UAVs for data transmission simultaneously. Therefore, the BS will allocate the subchannels to UAVs according to the following mechanism.

In each frame, the BS schedules $K$ subchannels to maximize the sum of STPs of all UAVs. Here, the STP for the $i$-th UAV is defined by $\mathcal{P}_{ts,i} = 1-(1-\mathcal{P}^{u}_{ts,i})(1-\mathcal{P}^{c}_{ts,i})$, i.e., the probability that the sensory data can be transmitted successfully in either the U2D or the cellular mode. Besides, when a UAV has already successfully transmitted its sensory data before a frame, the BS is without the need of allocating subchannel to this UAV in the current and future frames any more. Since different UAVs are required to utilize orthogonal subchannels, it is equivalent to that the BS allocates $K$ subchannels to the UAVs with the highest $K$ STPs which have not successfully transmitted their data.

More explicitly, we denote the transmission state indicators of the UAVs in the $t$-th frame of the $c$-th cycle as $\textbf{\emph{I}}^{(c)}(t) = \left( I_1^{(c)}(t), I_2^{(c)}(t),...,I_N^{(c)}(t) \right)$, in which the value of $I_i^{(c)}(t)$ is given by
\begin{equation}
I_i^{(c)}(t)=
\left\{
\begin{aligned}
0,&\quad \text{if the transmission has failed for the $i$-th UAV},\\
1,&\quad \text{if the U2D transmission has succeeded for the $i$-th UAV},\\
2,&\quad \text{if the cellular transmission has succeeded for the $i$-th UAV}.
\end{aligned}
\right.
\end{equation}
Based on the above notations, we introduce the subchannel allocation vector in the $t$-th frame of the $c$-th cycle as $\textbf{\emph{v}}^{(c)}(t) = \left( v_1^{(c)}(t), v_2^{(c)}(t),...,v_N^{(c)}(t) \right)$, where
\begin{equation}\label{subchannel allocation vector}
v_i^{(c)}(t)=
\left\{
\begin{aligned}
1,&\quad {\mathcal{P}^{(c)}_{ts,i}(t) \left[1-\textmd{sgn}\left(I^{(c)}_i(t)\right)\right] \geq \left\{ \emph{\textbf{P}}^{(c)}_{ts}(t) \left[\textbf{1}-\textmd{sgn}\left(\emph{\textbf{I}}^{(c)}(t)\right)\right] \right\}_K  },\\
0,&\quad {\textmd{otherwise}}.\\
\end{aligned}
\right.
\end{equation}
Here, $\textmd{sgn}(\cdot)$ is the sign function, $\mathcal{P}^{(c)}_{ts,i}(t)$ denotes the STP for the $i$-th UAV in the $t$-th frame of the $c$-th cycle, and $\left\{ \emph{\textbf{P}}^{(c)}_{ts}(t) \left[\textbf{1}-\textmd{sgn}\left(\emph{\textbf{I}}^{(c)}(t)\right)\right] \right\}_K $ denotes the $K$-th largest STPs among the UAVs which have not succeeded in sending the sensory data before the $t$-th frame.

Under the above subchannel allocation mechanism, each of the UAVs not succeeding in data transmission has the incentive to compete with others by designing its trajectory and selecting its transmission mode. It is worth mentioning that the subchannel allocation result are determined by the \emph{ranking} of all UAVs' STPs, rather than the \emph{absolute values} of them. Therefore, each UAV should take other UAVs' locations into consideration when it optimizes its trajectory and transmission mode.

\section{Joint Sensing and Transmission Analysis}
%%%%%%%%%%%%%%%%%%%%%%%%%%%%%%%%%%%
\label{Markov Analysis}
%%%%%%%%%%%%%%%%%%%%%%%%%%%%%%%%%%%

In this section, we analyze the joint sensing and transmission protocol by specifying the state transitions of UAVs in nested bi-level Markov chains. Specifically, the outer and inner Markov chains depict the state transitions in the sensing and transmission processes, respectively. For simplicity, we omit the superscript cycle index $c$ in the following notations.

\subsection{Outer Markov Chain for UAV Sensing}
\begin{figure}[!t]
\centering
\vspace{-5mm}
\includegraphics[width=6in]{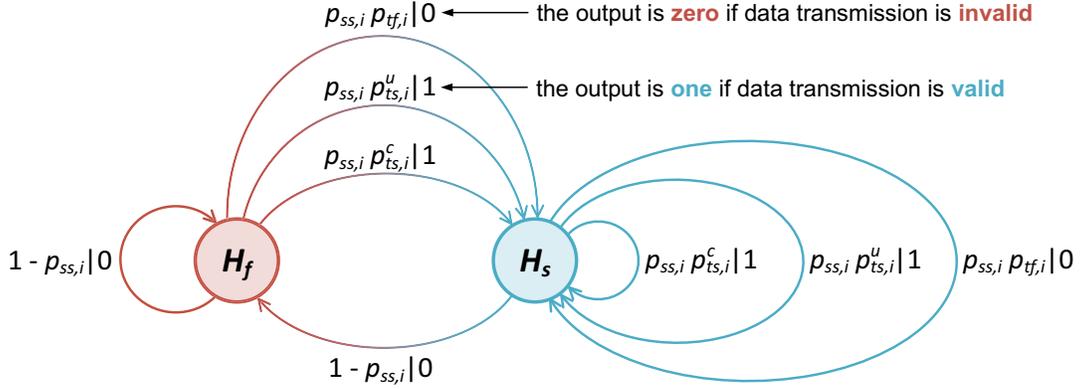}
\vspace{-10mm}
\caption{The outer Markov chain.}
\vspace{-7mm}
\label{outer_Markov_Chain}
\end{figure}

The sensing part in the joint sensing and transmission protocol can be evaluated by an outer Markov chain, where state transitions take place in each cycle. As shown in Fig.~\ref{outer_Markov_Chain}, a UAV has two sensing states in each cycle, i.e., state $H_f$ denoting that the sensing fails, and state $H_s$ denoting that the sensing succeeds. Define the SSP of the $i$-th UAV in a cycle as $p_{ss,i}$. Then, the UAV's state transits to $H_s$ with probability $p_{ss,i}$, and transits to $H_f$ with probability $1-p_{ss,i}$.

Since the location of a UAV changes slightly during each frame, we assume that the location of each UAV is fixed within each frame. Then, we can express the location of the $i$-th UAV in a cycle as $\textbf{\emph{x}}_i(t) = \left(x_i(t), y_i(t), h_i(t)\right)$, $t\in[1,T_c]$. Likewise, the distance between the $i$-th UAV and its target in a cycle can be expressed as $l_i(t) = \|\textbf{\emph{x}}_i(t) - \textbf{\emph{x}}^t_i\|_2$, $t\in[1,T_c]$. Moreover, we assume that each UAV moves with a uniform speed and direction in each cycle. Based on above assumptions, the location of the $i$-th UAV at the $t$-th frame in a cycle can be given by
\begin{equation}
\textbf{\emph{x}}_i(t) = \textbf{\emph{x}}'_i(T_c) + \frac{t}{T_c} [ \textbf{\emph{x}}_i(T_c) - \textbf{\emph{x}}'_i(T_c) ], \quad t\in[1,T_c].
\end{equation}
where $\textbf{\emph{x}}'_i(T_c)$ denotes the location of the $i$-th UAV at the last frame in the last cycle. According to the protocol in Section~\ref{Sense-and-Send Protocol}, a UAV is required to continuously sense its target for $T_s$ frames within a cycle. Thus, we can calculate the SSP of the $i$-th UAV in a cycle, i.e., $p_{ss,i}$, by
\begin{equation}
p_{ss,i} = \prod_{t=1}^{T_s} \mathcal{P}_{ss,i}(t) =  \prod_{t=1}^{T_s} e^{-\lambda t_f l_i(t)}.
\end{equation}

After sensing the target, a UAV may have three possible situations at the end of the transmission part: 1)~The transmission has failed; 2)~The U2D transmission has succeeded; 3)~The cellular transmission has succeeded. The probabilities of these three situations for the $i$-th UAV are denoted by $p_{tf,i}$, $p^u_{ts,i}$, and $p^c_{ts,i}$, respectively. The values of these probabilities can be calculated by the inner Markov chain, which will be introduced in the next subsection.

We define the data transmission of a UAV is \emph{valid} when the UAV successfully senses its target as well as transmits the sensory data to its mobile device. Then, we can express the valid transmission probability~(VTP) for the $i$-th UAV in the U2D and the cellular modes by $p_{ss,i}p^u_{ts,i}$ and $p_{ss,i}p^c_{ts,i}$, respectively. Besides, the probability that the $i$-th UAV successfully senses its target but fails to transmit the sensory data to its mobile device can be given by $p_{ss,i}p_{tf,i}$. To characterize whether a transmission is valid or not, we define the output of each state transition, which is denoted on the right side of each transition probability in Fig.~\ref{outer_Markov_Chain}. The output is one when the data transmission is valid, otherwise, the output is zero.

\subsection{Inner Markov Chain for UAV Transmission}

In this part, we introduce an inner Markov chain to formulate the transmission part in the joint sensing and transmission protocol. Since the general state transmission diagram is complicated, we illustrate the inner Markov chain via an example with the number of subchannels $K = 1$, the number of UAVs $N = 3$, and the number of frames in each transmission part $T_t = 2$. We take the $1$-st UAV as an example and illustrate its state transition diagram in Fig.~\ref{inner_Markov_Chain}.

\begin{figure}[!t]
\centering
\vspace{-3mm}
\includegraphics[width=6in]{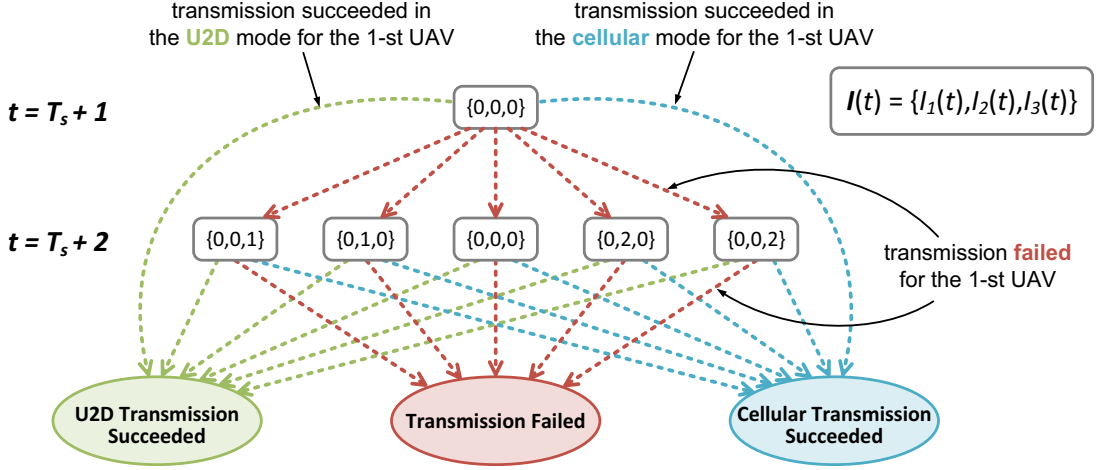}
\vspace{-7mm}
\caption{The inner Markov chain of the $1$-st UAV with $K = 1$, $N = 3$ and $T_t = 2$.}
\vspace{-7mm}
\label{inner_Markov_Chain}
\end{figure}

As shown in Fig.~\ref{inner_Markov_Chain}, in the first frame of the transmission part, i.e., $t = T_s + 1$, none of UAVs has transmitted its sensory data, and thus, we have $\textbf{\emph{I}}(T_s + 1) = \{0,0,0\}$. During this frame, all UAVs should compete with each other for subchannels. In the next frame, if the $1$-st UAV successfully transmits its sensory data in the U2D or the cellular mode, the state will transit to \emph{U2D Transmission Succeeded} or \emph{Cellular Transmission Succeeded}, with probability $\mathbb{P}^u_{ts,1}(T_s + 1)v_1(T_s + 1)$ or $\mathbb{P}^c_{ts,1}(T_s + 1)v_1(T_s + 1)$, accordingly. Here, $v_1(\cdot)$ is the subchannel allocation vector for the $1$-st UAV, which is given in~(\ref{subchannel allocation vector}). Besides, $\mathbb{P}^u_{ts,1}(\cdot)$ and $\mathbb{P}^c_{ts,1}(\cdot)$ can be calculated by the following proposition.

\begin{proposition}\label{proposition1}
When the $i$-th UAV is allocated to a subchannel, it will successfully transmit its sensory data in the U2D or the cellular mode with probability $\mathbb{P}^u_{ts,i}$ or $\mathbb{P}^c_{ts,i}$, in which
\begin{align}
\mathbb{P}^u_{ts,i} &= \mathcal{P}^{u}_{ts,i}\left(1-\mathcal{P}^{c}_{ts,i}\right) + \mathcal{P}^{u}_{ts,i} \mathcal{P}^{c}_{ts,i} \cdot \mathcal{P}\left\{R^u_{i} \geq R^c_{i} \right\}, \\
\mathbb{P}^c_{ts,i} &= \left(1-\mathcal{P}^{u}_{ts,i}\right)\mathcal{P}^{c}_{ts,i} + \mathcal{P}^{u}_{ts,i} \mathcal{P}^{c}_{ts,i} \cdot \mathcal{P}\left\{R^u_{i} < R^c_{i} \right\},
\end{align}
where $\mathcal{P}\left\{R^u_{i} \geq R^c_{i} \right\} =  1-\int_{-\infty}^\infty \Psi_i(y) f_{R_i^u}\left( y \right)dy$, $\mathcal{P}\left\{R^u_{i} < R^c_{i} \right\} =  \int_{-\infty}^\infty \Psi_i(y) f_{R_i^u}\left( y \right)dy$, with $\Psi_i(y) = \mathcal{P}_{a,i}(\textbf{x}_0, \textbf{x}_i, 2y) \mathcal{P}_{g,i}(\textbf{x}^d_i,\textbf{x}_0,2y)$, $f_{R_i^u}\left( y \right) = \mathcal{P}_{LoS,i} \cdot f_{R_{LoS,i}^u}\left( y \right) + \mathcal{P}_{NLoS,i} \cdot f_{R_{NLoS,i}^u}\left( y \right)$, $f_{R_{LoS,i}^u}\left( y \right) = \frac{2\ln2(K_{ri}+1)(2^y - 1)2^y }{\xi^2_{LoS,i}} \exp\left\{ -\frac{(K_{ri}+1)(2^y - 1)^2 + K_{ri}\xi_{LoS,i}^2}{\xi_{LoS,i}^2} \right\}\cdot
I_0 \left( \frac{ 2\sqrt{(K_{ri}+1)K_{ri}} \left(2^y - 1\right)}{\xi _{LoS,i}}\right)$, \\ $f_{R_{NLoS,i}^u}\left( y \right) = \frac{\ln2\left(2^y - 1\right)2^y}{\xi^2_{NLoS,i}} \exp\left\{ -\frac{\left(2^y - 1\right)^2}{2\xi^2_{NLoS,i}} \right\}$, $\xi_{LoS,i}= \frac{P^u}{N_0 \cdot 10^{\mathcal {L}_{LoS,i}/10}}$, and $\xi_{NLoS,i}=\frac{P^u}{N_0 \cdot 10^{\mathcal {L}_{NLoS,i}/10}}$.
\end{proposition}
\begin{proof}
See Appendix~\ref{Proof1}.
\end{proof}

However, if the $1$-st UAV fails to transmit the sensory data to its mobile device in the first frame,  $\textbf{\emph{I}}(t)$ will transit into other states in the second frame, which is determined by the transmission states of other UAVs. For instance, when the U2D transmission for the $2$-nd UAV succeeds in the first frame, we have $\textbf{\emph{I}}(T_s + 2) = \{0,1,0\}$. As the $2$-nd UAV has succeeded in data transmission, it will quit the competition for subchannels in the following frames. Consequently, the $1$-st UAV will only face one competitor, namely the $3$-rd UAV, in the second frame. In other words, the $1$-st UAV will have a larger probability to win a subchannel than before. Finally, in the last frame of the cycle, namely $t = T_c$, the transmission state of the $1$-st UAV will transit into \emph{Transmission Failed} if it does not transmit the sensory data successfully.

Based on the above assumptions, we can obtain the values of $p_{tf,i}$, $p^u_{ts,i}$ and $p^c_{ts,i}$ in the outer Markov chain by calculating the absorbing probabilities of the states \emph{Transmission Failed}, \emph{U2D Transmission Succeeded}, and \emph{Cellular Transmission Succeeded} in the inner Markov chain, respectively. Since the absorbing probabilities are difficult to be derived, we adopt the recursive algorithm in~\cite{JHL-2018} to obtain the numerical results.

Specifically, given the state at the $t$-th frame $\emph{\textbf{I}}(t)$, $t\in[T_s+1,T_c]$, we define the probabilities that the $i$-th UAV successfully transmits its sensory data \emph{after} the $t$-th frame as $\mathbb{P}^u_{ts,i}\{t|\emph{\textbf{I}}(t)\}$ and $\mathbb{P}^c_{ts,i}\{t|\emph{\textbf{I}}(t)\}$, respectively. For these two probabilities, the following two equations hold:
\begin{equation}\label{Pr_tu i}
\mathbb{P}^u_{ts,i}\{t|\emph{\textbf{I}}(t)\} = \mathbb{P}^u_{ts,i}(t)v_i(t) + \sum_{\emph{\textbf{I}}(t+1), \emph{I}_i(t+1)=0}
\mathcal{P}\{\emph{\textbf{I}}(t+1)|\emph{\textbf{I}}(t)\} \mathbb{P}^u_{ts,i}\{t+1|\emph{\textbf{I}}(t+1)\},
\end{equation}
\begin{equation}\label{Pr_tc i}
\mathbb{P}^c_{ts,i}\{t|\emph{\textbf{I}}(t)\} = \mathbb{P}^c_{ts,i}(t)v_i(t) + \sum_{\emph{\textbf{I}}(t+1), \emph{I}_i(t+1)=0}
\mathcal{P}\{\emph{\textbf{I}}(t+1)|\emph{\textbf{I}}(t)\} \mathbb{P}^c_{ts,i}\{t+1|\emph{\textbf{I}}(t+1)\}.
\end{equation}
Here, the first terms in the righthand of equations~(13) and~(14) indicate the probabilities that the data transmission of the $i$-th UAV succeeds within the $t$-th frame in the U2D and the cellular modes, respectively. Besides, the second terms in the righthand of equations~(13) and~(14) accordingly imply the probabilities that the data transmission of the $i$-th UAV fails within the $t$-th frame and succeeds after the $(t+1)$-th frame in the U2D and the cellular modes, in which $\mathcal{P}\{\emph{\textbf{I}}(t+1)|\emph{\textbf{I}}(t)\}$ denotes the state transition probabilities from state $\emph{\textbf{I}}(t)$ to state $\emph{\textbf{I}}(t+1)$.

\begin{table}[!t]
\centering
\caption{Transition Probability of the $i$-th UAV in the $t$-th Frame}
\vspace{-3mm}
\begin{tabular}{|p{2in}<{\centering}|p{0.5in}<{\centering}|p{0.5in}<{\centering}|p{0.5in}<{\centering}|}
  \hline
  \diagbox{$I_i(t)$}{$\mathcal{P}\left\{I_i(t+1)|I_i(t)\right\}$}{$I_i(t+1)$} & 0 & 1 & 2 \\
  \hline
  0 & $\mathbb{P}_{tf,i}$ & $\mathbb{P}^u_{ts,i}$ & $\mathbb{P}^c_{ts,i}$ \\
  \hline
  1 & 0 & 1 & 0 \\
  \hline
  2 & 0 & 0 & 1 \\
  \hline
\end{tabular}
\vspace{-7mm}
\label{Transition Probability Table}
\end{table}

Since the STPs of different UAVs are independent with each other, we have $\mathcal{P}\{\emph{\textbf{I}}(t+1)|\emph{\textbf{I}}(t)\} = \prod_{i=1}^N \mathcal{P}\{\emph{I}_i(t+1)|\emph{I}_i(t)\}$, where the values of $\mathcal{P}\{\emph{I}_i(t+1)|\emph{I}_i(t)\}$ are shown in Table~\ref{Transition Probability Table}. In the first row of Table~\ref{Transition Probability Table}, $\mathbb{P}^u_{ts,i}$ and $\mathbb{P}^c_{ts,i}$ are given in Proposition~\ref{proposition1}, and $\mathbb{P}_{tf,i} = (1-\mathcal{P}^u_{ts,i})(1-\mathcal{P}^c_{ts,i})$ denotes the probability that the data transmission of the $i$-th UAV fails. Besides, the second and third rows of Table~\ref{Transition Probability Table} indicate that the transmission state of a UAV will keep unchanged when it has already successfully transmitted the sensory data. According to equations~(\ref{Pr_tu i}) and~(\ref{Pr_tc i}), we can solve $\mathbb{P}^u_{ts,i}\{t|\emph{\textbf{I}}(t)\}$ and $\mathbb{P}^c_{ts,i}\{t|\emph{\textbf{I}}(t)\}$ via the recursive algorithm in~\cite{JHL-2018}, respectively. Thus, the probabilities for the $i$-th UAV to successfully transmit its sensory data in the U2D and the cellular modes within the transmission part can be given by $p^u_{ts,i} = \mathbb{P}^u_{ts,i}\{T_s+1|\emph{\textbf{I}}(T_s+1)\}$ and $p^c_{ts,i} = \mathbb{P}^c_{ts,i}\{T_s+1|\emph{\textbf{I}}(T_s+1)\}$, respectively. Then, the probability that the $i$-th UAV fails to transmit the sensory data within the transmission part is calculated by $p_{tf,i} = 1 - p^u_{ts,i} - p^c_{ts,i}$.

\section{Trajectory Design using Multi-agent Deep Reinforcement Learning}
%%%%%%%%%%%%%%%%%%%%%%%%%%%%%%%%%%%
\label{Trajectory Design using DRL}
%%%%%%%%%%%%%%%%%%%%%%%%%%%%%%%%%%%
In the system, UAVs aim to achieve valid data transmissions, i.e., both the sensing
and the transmission processes are successful. Since UAVs' sensing and
transmission processes have already been formulated as the state transitions in the nested bi-level Markov chains, UAVs can design their trajectories in consideration of their
sensing and transmission states at the same time. In this section, we investigate on the trajectory design problem for UAVs using multi-agent DRL, and propose a DQN-based multi-UAV trajectory design algorithm to solve this problem. Before that, we first discrete the space to design the trajectories for UAVs.

\subsection{Space Discretization}

\begin{figure}[!t]
\centering
\vspace{-12mm}
\includegraphics[width=6in]{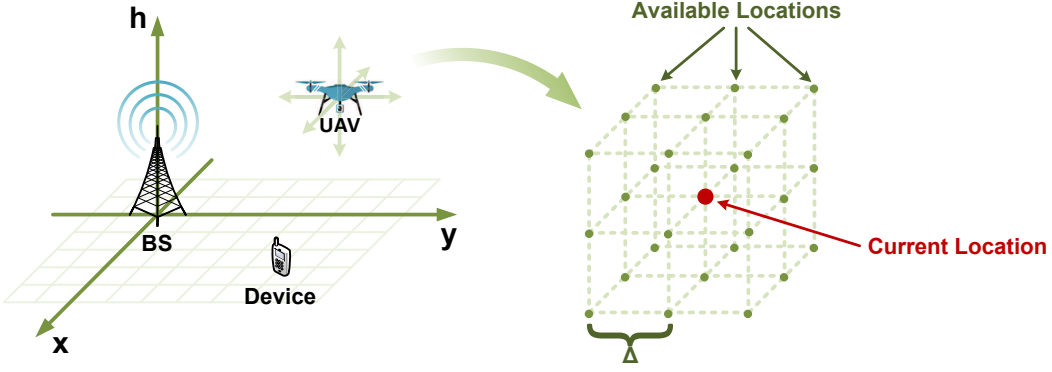}
\vspace{-9mm}
\caption{Lattice model of available locations that a UAV can select in the next cycle.}
\vspace{-7mm}
\label{lattice}
\end{figure}

We introduce a lattice model to describe the movements of UAVs in Fig.~\ref{lattice}. We assume that the space is divided into a finite set of discrete spatial points $\mathcal {S}$, which can be illustrated as a square lattice pattern. Thus, the trajectory of the $i$-th UAV starting from the $c$-th cycle can be represented as the sequence of locations $\mathcal{T}_i^{(c)} = \left\{\emph{\textbf{x}}_i^{(c)}, \emph{\textbf{x}}_i^{(c+1)}, ...\right\}$, in which $\emph{\textbf{x}}_i^{(c)}$ denote the location of the $i$-th UAV at the beginning of the $c$-th cycle, with $\emph{\textbf{x}}_i^{(c)} \in \mathcal {S}$, and $c\in[1,\infty)$.

As shown in Fig.~\ref{lattice}, the set of available locations that a UAV can select in the $c$-th cycle is a lattice with the center of $\emph{\textbf{x}}_i^{(c)}$, consisting of at most $27$ available locations. Besides, we define distance between any two adjacent available locations as $\Delta$, and thus, the maximum distance that a UAV can move in a cycle is $\sqrt{3}\Delta$. The value of $\Delta$ is determined by the maximum flying speed of a UAV, denoted by $v_{max}$. Specifically, we have $\Delta = T_c t_f v_{max}/\sqrt{3}$, in which $T_c$ is the number of frames within a cycle, and $t_f$ is the duration of a frame. Further, we introduce the available action set $\mathcal {A}(\emph{\textbf{x}})$ to denote the set of all vectors from the current location $\textbf{\emph{x}}$ to all available locations in the next cycle.

\subsection{Problem Formulation}

In this paper, we define the \emph{utility} of a UAV as the total number of valid data transmissions. Each UAV aims to maximize its utility by designing its trajectory. Besides, considering the timeliness requirements of sensing tasks, we introduce the discounting valuation on the successfully transmitted valid data. To be specific, for UAVs in the $c$-th cycle, the successfully transmitted valid data in the $c'$-th cycle ($c'>c$) is worth only $\rho^{c'-c}$ of that in the current cycle. Here, $\rho$ is the discount factor and $\rho \in [0,1]$. Therefore, the utility of the $i$-th UAV at the $c$-th cycle is defined as the total discounted rewards in the future, which is given by $\mathcal {U}_i^{(c)} = \sum_{n=0}^{\infty} \rho^{n} r_i^{(c+n)}$, in which $r_i^{(c)}$ denotes the reward of the $i$-th UAV in the $c$-th cycle. We define the reward of the $i$-th UAV in the $c$-th cycle as its sum VTP in both of the transmission modes within this cycle, i.e., $r_i^{(c)}  = p_{ss,i}p^u_{ts,i} + p_{ss,i}p^c_{ts,i}$, which is consistent with the output of the outer Markov chain in Fig.~\ref{outer_Markov_Chain}. Based on the above assumptions, the trajectory design problem in the UAV system can be formulated by
\begin{equation}
\begin{split}
\mathop {\max } \limits_{ \mathcal{T}_i^{(c)} } ~~& \mathcal {U}_i^{(c)} \\
s.t. ~~& \textbf{\emph{x}}_i^{(c'+1)} - \textbf{\emph{x}}_i^{(c')} \in \mathcal {A} \left(\textbf{\emph{x}}_i^{(c')}\right), c' \in [c,\infty).
\end{split}
\label{trajectory design problem}
\end{equation}

\subsection{Reinforcement Learning Formulation}

Since the rewards of UAVs in the future cycles are determined by the trajectories of all UAVs, the trajectory design problem~(\ref{trajectory design problem}) is difficult to solve. Fortunately, this problem can be regarded as a MDP~\cite{RA-1998}, and thus, we can solve this problem using reinforcement learning~(RL). Under the RL framework, each UAV is regarded as an \emph{agent}. The states and the actions of each agent are defined as the locations and the movements of the corresponding UAV, respectively. Based on~\cite{KMMA-2017}, we can characterize UAVs by a tuple $ < \mathcal {S}, \mathcal {A}_1,...,\mathcal {A}_N, \mathcal {P}, \mathcal {R}_1,...,\mathcal {R}_N, \rho >$, where:
%In what follows, we consider a general cycle and omit superscript cycle index $c$ for simplicity.
\begin{itemize}
  \item $\mathcal {S}$ denotes the state space including all possible states of UAVs at the beginning of each cycle. The state of the $i$-th UAV, denoted by $\textbf{\emph{s}}_i$, is described by its location, i.e., $= \textbf{\emph{x}}_i$;
  \item $\mathcal {A}_i, i\in\mathcal {N}$ denotes the action space consisting of all available actions of the $i$-th UAV at the each cycle. The action of the $i$-th UAV, denoted by $\emph{\textbf{a}}_i$, is defined by its movement at a given state\footnote{According to the lattice model shown in Fig.~\ref{lattice}, a UAV may have at most $27$ available movements at a state, and thus, the cardinality of the action space is up to $27$.}. At state $\emph{\textbf{s}}_i$, the available action set of the $i$-th UAV is expressed as $\mathcal {A}_i = \mathcal {A}(\emph{\textbf{s}}_i)$;
  \item $\mathcal {P}: \mathcal {S}^N \times \prod_{i=1}^N\mathcal {A}_i \rightarrow \Pi(\mathcal {S})$ is the state transition function, which maps the state spaces and the action spaces of all UAVs in the current cycle to their state spaces in the next cycle. Here, $\Pi(\mathcal {S})$ is the probability distribution over the state space $\mathcal {S}$;
  \item $\mathcal {R}_i:\mathcal {S}^N\times\mathcal {A}_i \rightarrow \Pi(0,1), i\in\mathcal {N}$ is the reward function of the $i$-th UAV, which maps the state spaces and the action spaces of the UAV in the current cycle to its reward. As the reward of each UAV in a cycle equals to its sum VTP in both of the transmission modes, the reward of the $i$-th UAV taking action $\emph{\textbf{a}}_i$ at the state $\emph{\textbf{s}}$ is given by $\mathcal {R}_i(\emph{\textbf{s}}, \emph{\textbf{a}}_i) = p_{ss,i}p^u_{ts,i} + p_{ss,i}p^c_{ts,i}$;
  \item $\rho \in [0,1]$ represents the discount factor, which evaluates the successfully transmitted valid data of UAVs in the future.
\end{itemize}

At the beginning of each cycle, the $i$-th UAV observes the state $\emph{\textbf{s}} = \{\textbf{\emph{s}}_i \}$, i.e., the states of all UAVs\footnote{Based on the protocol in Section~\ref{Sense-and-Send Protocol}, the BS is required to inform each UAV of the locations of all UAVs at the end of the last cycle. Thus, all UAVs' locations can be observed by each UAV at the beginning of the current cycle.}, from state space $\mathcal {S}^N$. Then, it takes an action $\emph{\textbf{a}}_i$ from action space $\mathcal {A}_i$ based on its \emph{policy}, denoted by $\pi_i$. Here, the policy is a mapping from the state space to the action space, which can be expressed by $\emph{\textbf{a}}_i = \pi_i(\emph{\textbf{s}})$. As action $\emph{\textbf{a}}_i$ taken by the $i$-th UAV, it receives a reward $r_i = \mathcal {R}_i(\emph{\textbf{s}}, \emph{\textbf{a}}_i)$ and observes a new state $\emph{\textbf{s}}'$, namely the states of all UAVs in the next cycle. Therefore, the trajectory design problem~(\ref{trajectory design problem}) can be transformed into maximizing the total discounted rewards of all UAVs in the system by optimizing their policies.

\subsection{Algorithm Design}
In this part, we propose a DQN-based multi-UAV trajectory design algorithm to optimize the policies for all UAVs in the system. Specifically, our proposed algorithm adopts Q-learning for UAVs to learn how to take actions optimally at different states, in which the rewards of state-action pairs are estimated by DQNs.

Q-learning is a model-free RL algorithm which can effectively optimize the policies for agents~\cite{CP-1992,J-1994}. In our system, the Q-value of the $i$-th UAV, denoted by $Q_i \left( \emph{\textbf{s}},\emph{\textbf{a}}_i \right)$, is defined as the accumulated discounted rewards when it takes action $\emph{\textbf{a}}_i$ at state $\emph{\textbf{s}}$ and then follows its policy $\pi_i$ afterwards. In Q-learning, the $i$-th UAV at each state will select the action which maximizes its Q-value. Therefore, the policy of the $i$-th UAV at state $\emph{\textbf{s}}$, i.e., $\pi_i(\emph{\textbf{s}})$, is
\begin{equation}
\pi_i(\emph{\textbf{s}}) = \mathop {\arg \max }\limits_{\emph{\textbf{a}}_i' \in \mathcal {A}_i(\emph{\textbf{s}})} Q_i(\emph{\textbf{s}},\emph{\textbf{a}}_i').
\end{equation}
After taking the action $\emph{\textbf{a}}_i = \pi_i(\emph{\textbf{s}})$, the $i$-th UAV receives a reward $r_i=\mathcal {R}_i(\emph{\textbf{s}}, \emph{\textbf{a}}_i)$, and then observes the next state\footnote{Actually, a UAV does not know the actions of other UAVs in the current cycle by itself. However, it can still observe the next state according to the information from the BS.} $\emph{\textbf{s}}'$. The optimal Q-values of the $i$-th UAV can be obtained iteratively based on the following update rule:
\begin{equation}
Q_i(\emph{\textbf{s}},\emph{\textbf{a}}_i) \leftarrow Q_i(\emph{\textbf{s}},\emph{\textbf{a}}_i) + \alpha
\left(\mathcal {R}_i(\emph{\textbf{s}}, \emph{\textbf{a}}_i) + \rho \mathop {\max }\limits_{\emph{\textbf{a}}_i' \in \mathcal {A}_i(\emph{\textbf{s}}')} Q_i(\emph{\textbf{s}}',\emph{\textbf{a}}_i') - Q_i(\emph{\textbf{s}},\emph{\textbf{a}}_i) \right),
\end{equation}
in which $\alpha$ is the learning rate. Once the optimal Q-values of a UAV are obtained, the optimal policy is also determined.

When the state-action space is small, Q-learning can works efficiently by maintaining look-up tables for the update of Q-values~\cite{HG-2018}. However, in our system, the state-action space of the trajectory design problem is very large, which makes the update of look-up tables infeasible. Fortunately, motivated by the deep neural network~(DNN), we can estimate the Q-values by a DNN function approximators, which are known as DQNs~\cite{nature-2015}. To be specific, DQNs can address the sophisticate mappings between the states of UAVs and their corresponding Q-values based on a large amount of training data. In our system, when a DQN is well-trained for the $i$-th UAV, given the current state $\emph{\textbf{s}}$ as the input of the DQN, we can obtain the Q-values of the $i$-th UAV taking different actions, i.e., $Q_i(\emph{\textbf{s}},\emph{\textbf{a}}_i), \emph{\textbf{a}}_i \in \mathcal {A}_i(\emph{\textbf{s}})$, from the outputs of the DQN.

In order to train the DQN of the $i$-th UAV $Q_i$, we utilize a separate network $\widehat{Q}_i$, namely the target network~\cite{nature-2015}, to generate the target for training. Define the weights of $Q_i$ and $\widehat{Q}_i$ as $\Theta_i$ and $\Theta_i^-$, respectively. During the training of $Q_i$, we update its weight $\Theta_i$ by minimizing the following loss function:
\begin{equation}
Loss(\Theta_i) = \sum\limits_{m \in \mathcal {D}_i} {\left(y - Q_i(\emph{\textbf{s}},\emph{\textbf{a}}_i; \Theta_i )\right)}^2,
\end{equation}
where $y = r_i + \max_{\emph{\textbf{a}}_i'} \widehat{Q}_i(\emph{\textbf{s}}',\emph{\textbf{a}}_i'; \Theta_i^- )$ denotes the target for training, and $\mathcal {D}_i$ is the training set with size $D$. Each sample in the training set, denoted by $m = \{\emph{\textbf{s}},\emph{\textbf{a}}_i,r_i,\emph{\textbf{s}}'\}$, contains the current state $\textbf{\emph{s}}$, the action taken by the $i$-th UAV $\textbf{\emph{a}}_i$, the reward received by the $i$-th UAV $r_i$, and the next state $\textbf{\emph{s}}'$. However, we do not update the weight of the target network $\Theta_i^-$ simultaneously with $\Theta_i$. Instead, we set $\Theta_i^- = \Theta_i$ every $\tau$ steps of the updating for $\Theta_i$.

\begin{algorithm}[t]
  \caption{DQN-based multi-UAV trajectory design algorithm for the $i$-th UAV.}
  \label{Trajectory Design Algorithm}
  \begin{algorithmic}[1]
    \REQUIRE
	    The structure of the network $Q_i$ and its target network $\widehat Q_i$; Maximum number of cycles $C$;
        Initial state $\textbf{\emph{s}}$.
    \ENSURE
        Policy $\pi_i$.
    \STATE Initialize replay memory $\mathcal {M}_i$;
    \STATE Initialize network $Q_i$ with random weight $\Theta_i$;
    \STATE Initialize target network $\widehat Q_i$ with weight $\Theta^{-}_i = \Theta_i$;
    \FOR{$c = 1: C$ }
    \STATE With probability $1-\epsilon$ select action $\emph{\textbf{a}}_i = { \arg \max}_{\emph{\textbf{a}}_i'} Q_i(\emph{\textbf{s}},\emph{\textbf{a}}_i';\Theta_i)$; otherwise, select a random action from $\mathcal {A}_i(\emph{\textbf{s}})$;
    \STATE Take action $\emph{\textbf{a}}_i$, and then observe reward $r_i$ and next state $\emph{\textbf{s}}'$;
    \STATE Store the sample $\{\emph{\textbf{s}},\emph{\textbf{a}}_i,r_i,\emph{\textbf{s}}'\}$ into replay memory $\mathcal {M}_i$;
    \STATE Sample a mini-batch $\mathcal {D}_i$ from replay memory $\mathcal {M}_i$;
    \STATE Perform a gradient descend step on the loss function $Loss(\Theta_i)$ with respect to weight $\Theta_i$ using data set $\mathcal {D}_i$ ;
    \STATE Set $\Theta_i^- = \Theta_i$ every $\tau$ steps;
	\ENDFOR
    \RETURN Policy $\pi_i$ based on the trained network $Q_i$.
	\end{algorithmic}
\end{algorithm}

The proposed DQN-based multi-UAV trajectory design algorithm is presented in Algorithm~\ref{Trajectory Design Algorithm}. In our algorithm, we adopt the \emph{experience replay} to suppress the temporal correlation in the generated training data~\cite{HG-2018}. Specifically, the training data of the $i$-th UAV is stored in a replay memory, denoted by $\mathcal {M}_i$, with size $M$. During the training of network $Q_i$, we randomly sample a mini-batch with $D$ samples from the replay memory $\mathcal {M}_i$ as the training set $\mathcal {D}_i$ in each iteration. Besides, we utilize $\epsilon$\emph{-greedy} policy to balance the exploration and exploitation~\cite{nature-2015,RTZ-2013}.
%Each agent takes action $\emph{\textbf{a}}_i = \pi_i(\emph{\textbf{s}})$ with probability $1-\epsilon$, and selects a random action from the available action set $\mathcal {A}_i(\emph{\textbf{s}})$ with probability $\epsilon$.
Moreover, the \emph{reduced available action set} method is used to accelerate the convergency~\cite{JHL-2018}.

\section{Performance Analysis}
%%%%%%%%%%%%%%%%%%%%%%%%%%%%%%%%%%%
\label{System Performance Analysis}
%%%%%%%%%%%%%%%%%%%%%%%%%%%%%%%%%%%
In this section, we first analyze the convergency and the complexity of our proposed algorithm. After that, we present some properties on the mode selection and the trajectory design for UAVs.

\subsection{Algorithm Analysis}
\subsubsection{Convergency}
In our algorithm, we adopt the gradient descend method to update the weight $\Theta_i$ for network $Q_i$, in which the learning rate exponentially decreases with iterations. Therefore, after a finite number of iterations, the weight $\Theta_i$ will converge to a certain value, which ensures the convergency of our proposed algorithm. Actually, the convergency of neural networks is challenging to be theoretically analyzed before training, as referred in~\cite{UWC-2019}. The analytical challenge lies in that the convergence of a neural network is highly dependent on the hyperparameters used during the training process, in which the quantitative relationship between the network convergency and the hyperparameters is complicated. Instead, in our paper, the convergency of our algorithm can be observed by simulation.

%With an appropriate step length, the value of the loss function $Loss(\Theta_i)$ is non-increasing after each iteration. Besides, denote the optimal weight of network $Q_i$ as $\Theta_i^*$, and thus, we have $Loss(\Theta_i) \geq Loss(\Theta_i^*)$ for all $\Theta_i$. Since the loss function $Loss(\Theta_i)$ is non-increasing with $\Theta_i$, and lower bounded with a finite value $Loss(\Theta_i^*)$, the convergency of our proposed algorithm is ensured.

\subsubsection{Complexity}
The time complexity of a network $Q_i$ is characterized by the number of operations in each iteration during the update of its weight $\Theta_i$. Assume the network $Q_i$ has $Z$ layers, whose numbers of neurons are denoted by $q_i$, $i=1,...,Z$. Then, the time complexity in each iteration can be expressed as $\mathcal {O}\left(\sum_{i=1}^{Z-1} q_i q_{i+1}\right)$. When all layers in the network have the same amount of neurons, denoted by $q$, the time complexity can be given by $\mathcal {O}\left((Z-1) q^2\right) = \mathcal {O}\left(q^2\right)$.

\subsection{UAV Performance Analysis}
\subsubsection{Mode Selection}
In the system, which transmission mode that a UAV may select is mainly determined by the location of this UAV. To analyze the mode selection of a UAV, we consider a general case that a UAV flies directly from its mobile device to the BS at a fixed altitude\footnote{The following analysis can also be applied to the cases where UAV's trajectory is in parallel to the direction from the BS and the mobile device.}, which is illustrated in Fig.~\ref{Remark_transmission}(a). The locations of the UAV, the mobile device, and the BS are given by $\emph{\textbf{x}}$, $\emph{\textbf{x}}^d$, and $\emph{\textbf{x}}_0$, respectively. Without loss of generality, we assume $\emph{\textbf{x}}^d=(x^d,0,0)$. Besides, we define the 2D distance from the BS to the UAV as $L_{2D}$, which is the distance from the mobile device to the UAV's projection on the ground. Moreover, we denote the STP of the UAV in the U2D and the cellular modes as $\mathcal{P}^{u}_{ts}$ and $\mathcal{P}^{c}_{ts}$, respectively. Based on these notations, we can show how the UAV-BS distance influences the UAV's STPs in different modes by Proposition~\ref{Remark1}.

\begin{figure}[!t]
\centering
\vspace{-5mm}
\includegraphics[width=6in]{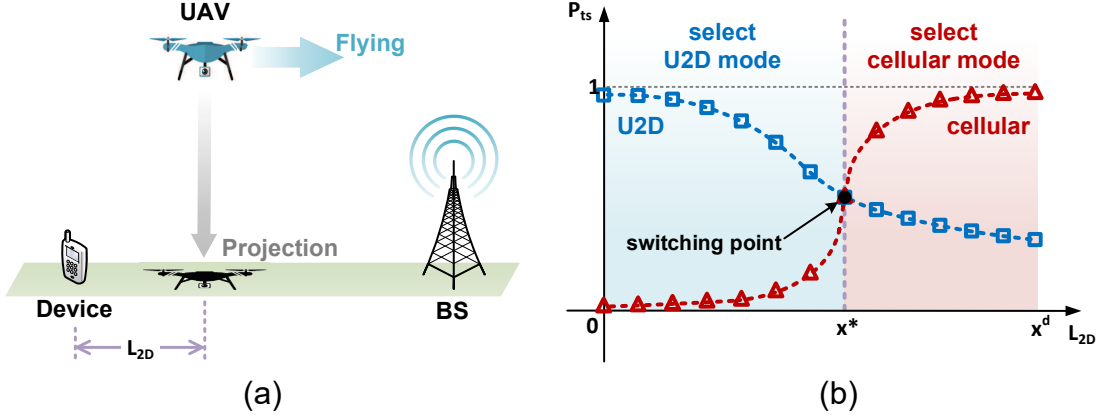}
\vspace{-12mm}
\caption{Mode selection for UAV transmissions.}
\vspace{-7mm}
\label{Remark_transmission}
\end{figure}

\begin{proposition}\label{Remark1}
  As the increase of $L_{2D}$, $\mathcal{P}^{u}_{ts}$ monotonously decreases while $\mathcal{P}^{c}_{ts}$  monotonously increases. Besides, $\mathcal{P}^{c}_{ts}$ changes more sharply with $L_{2D}$ than $\mathcal{P}^{u}_{ts}$.
\end{proposition}
\begin{proof}
According to subsection~\ref{Transmission Model}, given the QoS requirement $R_{th}$, we have $\mathcal{P}^{u}_{ts} = \mathcal{P}_{a,i}(\textbf{\emph{x}}^d,\textbf{\emph{x}},R_{th})$ and $\mathcal{P}^{c}_{ts} = \mathcal{P}_{a,i}(\textbf{\emph{x}}_0, \textbf{\emph{x}}, 2R_{th})\mathcal{P}_{g,i}(\textbf{\emph{x}}^d,\textbf{\emph{x}}_0,2R_{th})$. Since the locations of the mobile device and the BS are fixed, the second term in the expression of $\mathcal{P}^{c}_{ts}$, i.e., $\mathcal{P}_{g,i}(\cdot)$, is a constant. With the assumption that the UAV flies at a fixed altitude, as the increase of $L_{2D}$, $||\emph{\textbf{x}}-\emph{\textbf{x}}^d||_2$ increases, and thus, $\mathcal{P}^{u}_{ts}$ monotonously decreases. On the other side, $\mathcal{P}^{c}_{ts}$ monotonously increases with $L_{2D}$, since a higher value of $L_{2D}$ results in a lower value of $||\emph{\textbf{x}}-\emph{\textbf{x}}_0||_2$. Based on the above analysis, we can illustrate the STPs in both the transmission modes in Fig.~\ref{Remark_transmission}(b).

Furthermore, due to the TDM property of the cellular mode, the data transmissions in the cellular mode is harder to satisfy the QoS requirement than that of the U2D mode, which can be characterized by the terms $R_{th}$ and $2R_{th}$ in the expressions of $\mathcal{P}^{u}_{ts}$ and $\mathcal{P}^{c}_{ts}$, respectively. Therefore, $\mathcal{P}^{c}_{ts}$ changes more sharply with $L_{2D}$ than $\mathcal{P}^{u}_{ts}$, i.e., the absolute value of the slope of the red curve~(cellular) is larger than that of the blue one~(U2D) in Fig.~\ref{Remark_transmission}(b).
\end{proof}

According to the protocol in Section~\ref{Sense-and-Send Protocol}, the UAV tends to select the transmission mode with higher throughput. However, in this case, the transmission mode with a higher STP is equivalent to the one with a higher throughput. Therefore, based on Proposition~\ref{Remark1}, we can further figure out how the UAV select its transmission mode as it flies from its mobile device to the BS by Proposition~\ref{Remark2}.

\begin{proposition}\label{Remark2}
As the increase of $L_{2D}$, the UAV first tends to select the U2D mode, after a switching point it prefers the cellular mode. The UAV locating at the switching point is more likely to fail in transmission than that close to the mobile device or the BS.
\end{proposition}
\begin{proof}
When the UAV is close to the mobile device and far from the BS, i.e., $L_{2D}\rightarrow0$, the data transmissions are more likely to succeed in the U2D mode, since $\mathcal{P}^{u}_{ts}$ is higher than $\mathcal{P}^{c}_{ts}$. Besides, when the UAV is close to the BS and far from the mobile device, i.e., $L_{2D}\rightarrow x^d$, $\mathcal{P}^{u}_{ts}$ can reach a high value while $\mathcal{P}^{c}_{ts}$ is relatively low, and thus, the cellular mode can ensure successful data transmissions. Therefore, as $L_{2D}$ increases from $0$ to $x^d$, the UAV first tends to select the U2D mode, after a \emph{switching point}, denoted by $x^*$, it prefers the cellular mode.

When the UAV locates at the switching point, neither the U2D nor the cellular mode can guarantee successful data transmissions. Thus, we can conclude that UAV locating at the switching point is more likely to fail in transmission than that close to the mobile device or the BS. Since $\mathcal{P}^{u}_{ts}$ and $\mathcal{P}^{c}_{ts}$ are continuous functions of $L_{2D}$, we can find the location of $x^*$ by setting the derivative of $\mathcal{P}^{u}_{ts} - \mathcal{P}^{c}_{ts}$ to zero. Although the analytical solution of $x^*$ is very hard to obtain, we can still get the numerical solution of $x^*$ by simulations.
\end{proof}

\subsubsection{Trajectory Design}
In the system, the trajectory of a UAV can not be characterized by specific system parameters. In other words, even all of the system parameters are given, the trajectory of this UAV can not be determined before simulation. However, we can still find the factors that may influence the trajectories of the UAV. Specifically, we consider the case that a UAV performs sensing task for $C$ cycles at a fixed altitude, whose starting point locates at the above of its mobile device. In this case, the trajectory of the UAV can be limited in finite area determined by the BS, its mobile device, and its sensing target, which is given by Proposition~\ref{Remark3}.

\begin{figure}[!t]
\centering
\vspace{-3mm}
\includegraphics[width=3.5in]{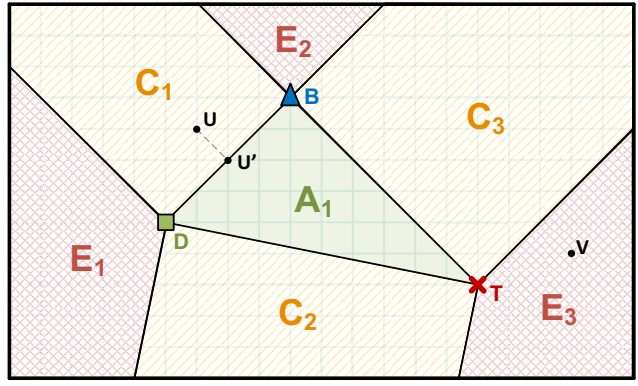}
\vspace{-7mm}
\caption{The trajectory of a UAV, where the BS, its mobile device, and its sensing target are denoted by points $B$, $D$, and $T$, respectively.}
\label{Remark_trajectory}
\vspace{-7mm}
\end{figure}

\begin{proposition}\label{Remark3}
In the top view of the system, the trajectory of a UAV is within a triangular area with the BS, its mobile device and its sensing target as the vertexes.
\end{proposition}

\begin{proof}
As illustrated in Fig.~\ref{Remark_trajectory}, when the BS, the mobile device and the sensing target are non-collinear, they divide the plane into $7$ areas, which can be classified in to three groups: $\mathcal {G}_1 = \left\{A_1\right\}$, $\mathcal {G}_2 = \left\{C_1, C_2, C_3 \right\}$, and $\mathcal {G}_3 = \left\{E_1, E_2, E_3 \right\}$. For any location in $\mathcal {G}_2$ or $\mathcal {G}_3$, we can find a location in $\mathcal {G}_1$ where the UAV can obtain a higher utility than that in $\mathcal {G}_2$ or $\mathcal {G}_3$. For instance, given a point $U$ within area $C_1$ in $\mathcal {G}_2$, we can find its foot point $U'$ on the segment $BD$, i.e., a side of the triangular area $A_1$. Since $U'$ is closer to the sensing target $T$ than $U$, the SSP at $U'$ is higher. Besides, $U'$ is closer to the BS $B$ as well as the mobile device $D$ than $U$, and thus, the STPs in both transmission modes at $U'$ are also higher. Therefore, the UAV locating at $U'$ can obtain a higher utility than that of $U$. Likewise, given a point $V$ within area $E_3$ in $\mathcal {G}_3$, we can conclude that UAV locating at the sensing target $T$, i.e., a vertex of the triangular area $A_1$, can obtain a higher utility than that of $V$ for the same reason.

Since the UAV aims to maximize its utility by designing its trajectory, we can conclude that it will not fly to the areas in $\mathcal {G}_2$ or $\mathcal {G}_3$. Therefore, trajectory of the UAV is within the triangular area $A_1$. When the BS, the mobile device and the sensing target are collinear, the trajectory of the UAV is within the longest segment among these three locations. This proposition can be further extended to the 3D case, where the trajectory of the UAV is within a triangular prism area, whose height is determined by the maximum and minimum flying altitudes of the UAV.
\end{proof}

In addition, since the UAV aims to maximize its utility by designing its trajectory, it is motivated to fly to the locations where it can achieve high utility. Then, the trajectory of the UAV can be further described by Proposition~\ref{Remark4}.

\begin{figure}[!t]
\centering
%\vspace{-1mm}
\includegraphics[width=6in]{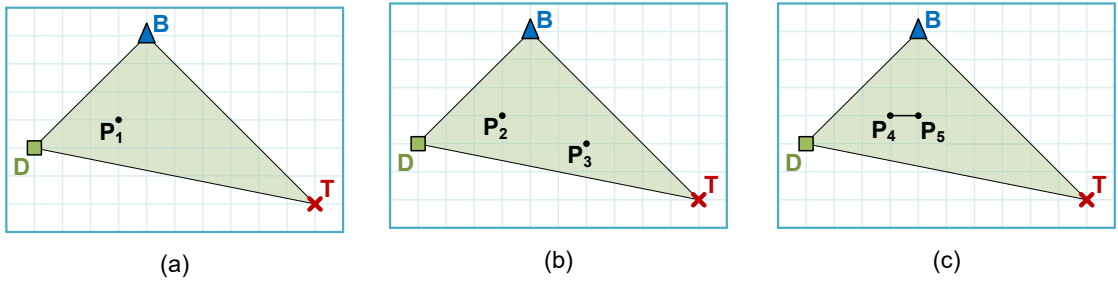}
\vspace{-11mm}
\caption{The optimal points of a UAV, where the BS, its mobile device, and its sensing target are denoted by points $B$, $D$ and $T$, respectively.}
\label{Remark_optimal}
\vspace{-7mm}
\end{figure}

\begin{proposition}\label{Remark4}
After performing enough cycles, i.e. $K\rightarrow\infty$, a UAV will keep still at a point or fly to and fro among multiple adjacent points.
\end{proposition}

\begin{proof}
Since the trajectory of a UAV is within a finite and discrete triangular area, there exists single or multiple optimal points in the area that the UAV can obtain the highest utility. Therefore, a UAV  aiming to maximize its utility has the tendency to fly to the optimal points. In the system, there may exist three cases for the optimal points:
\begin{itemize}
  \item \textbf{Single optimal point}: This case is shown in Fig.~\ref{Remark_optimal}(a). In this case, the UAV will fly to the optimal point $P_1$, and then keeps still.
  \item \textbf{Multiple non-adjacent optimal points}: This case is shown in Fig.~\ref{Remark_optimal}(b). In this case, the UAV flies to one of the optimal point $P_2$ or $P_3$, and then keeps still.
  \item \textbf{Multiple adjacent optimal points}: This case is shown in Fig.~\ref{Remark_optimal}(c). In this case, the UAV first flies to one of the optimal points $P_4$ or $P_5$, after which it may fly to and fro among these adjacent points.
\end{itemize}
Therefore, we can conclude that a UAV will keep still at a point or fly to and fro among multiple adjacent points after performing enough cycles.
\end{proof}

\section{Simulation Results}
%%%%%%%%%%%%%%%%%%%%%%%%%%%%%%%%%%%
\label{Simulation Results}
%%%%%%%%%%%%%%%%%%%%%%%%%%%%%%%%%%%

\begin{table}[!t]
\centering
\caption{Parameters for Simulation} \label{parameters}
\vspace{-3mm}
\begin{tabular}{|p{2.1in}|p{0.5in}|p{2.1in}|p{0.8in}|}
 \hline
 \textbf{Parameters} & \textbf{Values} & \textbf{Parameters} & \textbf{Values} \\
 \hline
 \hline
  Number of UAVs $N$ & 10 & Distance between two adjacent points $\Delta$ & 25 m\\
  \hline
  Number of subchannels $K$ & 2 & UAV sensing factor $\lambda$ & 0.0005 (s$\cdot$m)$^{-1}$ \\
  \hline
  Transmit power of UAVs $P^u$ & 10 dBm & Duration of a frame $t_f$ & 2 s\\
  \hline
  Transmit power of the BS $P^b$ & 41 dBm & QoS requirement $R_{th}$ & 1 $\textrm{bit}/\textrm{s}/\textrm{Hz}$\\
  \hline
  Noise power $N_0$ & -85 dBm & Length of the sensing part $T_s$ & 2\\
  \hline
  Carrier frequency $f_c$ & 2 GHz & Length of the transmission part $T_t$ & 2\\
  \hline
  Height of the BS $H_0$ & 10 m & Discount factor $\rho$ & 0.9\\
  \hline
  Minimum flying altitude of UAVs $h_{min}$ & 50 m & Size of mini-batch $D$ & 50\\
  \hline
  Maximum flying altitude of UAVs $h_{max}$ & 150 m & Size of replay memory $M$ & 3000\\
  \hline
  Maximum flying speed of UAVs $v_{max}$ & 5.4 m/s & Update frequency of target networks $\tau$ & 30\\
  \hline
\end{tabular}
\vspace{-5mm}
\end{table}

In this section, we present the simulation results on the trajectory design of UAVs in the U2D communications overlaying cellular network. The simulation parameters based on the existing 3rd Generation Partnership Project~(3GPP) technical reports~\cite{3GPP_TR_36_777} are given in Table~\ref{parameters}.

In the simulation, we model the cell as a circular area with the BS at the center and the radius as $500~\textrm{m}$. The sensing targets and mobile devices are randomly distributed within this circular area. The initial position of each UAV is on its corresponding mobile device with the altitude $100~\textrm{m}$. Furthermore, during the training of DQNs, the ReLU function, defined by $f_{ReLU}(x) = \textrm{max}\{0,x\}$, is adopted as the activation function~\cite{GTG-2013}. The learning rate has the initial value of $0.001$, and decreases exponentially with performing cycles. The exploration parameter $\epsilon$ in the $\epsilon$\emph{-greedy} policy decreases linearly during updating, whose initial and final value are $1.0$ and $0.1$, respectively.

\begin{figure}[!t]
\centering
\includegraphics[width=3.5in]{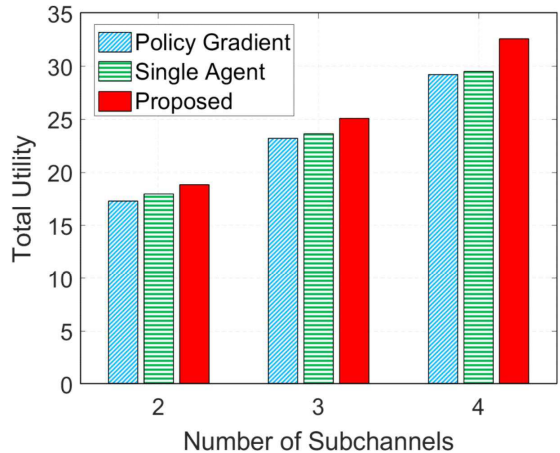}
\vspace{-8mm}
\caption{Performance comparison of three algorithms, given different number of subchannels $K$.}
\label{Simulation_Utility}
\vspace{-7mm}
\end{figure}

In Fig.~\ref{Simulation_Utility}, we compare our proposed algorithm with the following two algorithms:
\begin{itemize}
  \item \textbf{Policy gradient algorithm}~\cite{RDSY-1999}: Each UAV directly optimizes its parameterized control policy by a variant of gradient descent.
  \item \textbf{Single-agent Q-learning algorithm}~\cite{JHL-2018}: Each UAV updates its policy considering other UAVs as the environment.
\end{itemize}
To be specific, Fig.~\ref{Simulation_Utility} shows the performance comparison on the total utility in the system, given different numbers of subchannels $K$. We can observe that our proposed algorithm outperforms both policy gradient algorithm and single-agent algorithm in terms of the total utility. Besides, for any algorithm, the total utility increases with the number of subchannels $K$, since more frequency resources can be utilized by UAVs.

\begin{figure}[!t]
\centering
\includegraphics[width=6.4in]{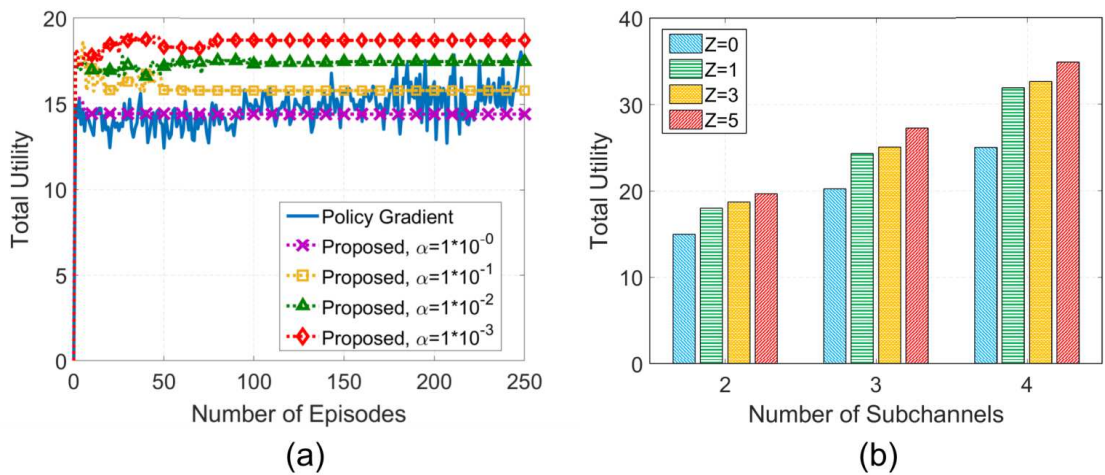}
\vspace{-10mm}
\caption{(a) Total utility versus the number of episodes, with different learning rate $\alpha$; (b) Performance comparison of our algorithm given different number of hidden layers $Z$. The DQN with $Z=0$ does not have hidden layer. The hidden layer of the DQN with $Z=1$ contains $250$ neurons. The hidden layers of the DQN with $Z=3$ have $500$, $250$, and $120$ neurons accordingly. The hidden layers of the DQN with $Z=5$ have $500$, $500$, $250$, $250$, and $120$ neurons accordingly. }
\label{Simulation_Convergency_Layer}
\vspace{-7mm}
\end{figure}

In Fig.~\ref{Simulation_Convergency_Layer}(a), we present the total utility versus the number of episodes, with different learning rate $\alpha$. Here, policy gradient algorithm is adopted as the benchmark. As shown in the figure, for all of the cases, our proposed algorithm converges after $100$ episodes, which is faster than the policy gradient algorithm. Besides, the network with a smaller learning rate may achieve a higher total utility in the system, while it may spend a little longer time to converge. In Fig.~\ref{Simulation_Convergency_Layer}(b), we compare the performance of our algorithm with different number of hidden layers $Z$. We can observe that the network with larger number of hidden layers can obtain higher total utility in the system. Likewise, for all cases, the total utility increases with the number of subchannels $K$, as more frequency resources can be utilized.

\begin{figure}[!t]
\centering
\vspace{-6mm}
\includegraphics[width=6in]{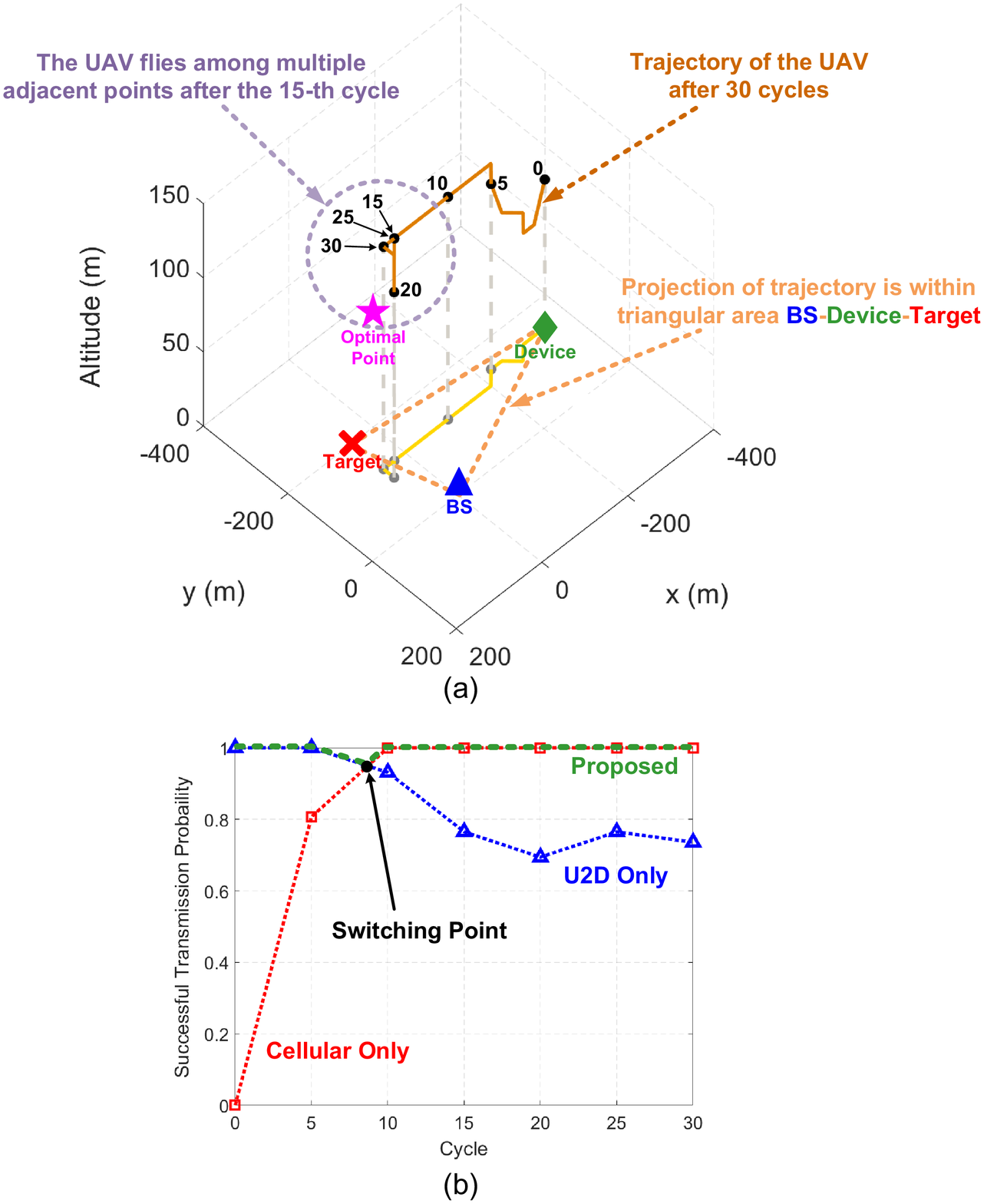}
\vspace{-12mm}
\caption{(a)~The trajectory of the UAV; (b)~The STPs of the UAV on its trajectory. The UAV's sensing target and mobile device are located at $\textbf{\emph{x}}^t = (50,-200,0)$ and $\textbf{\emph{x}}^d = (-350,-150,0)$, respectively. }
\label{Simulation_trajectory}
\vspace{-7mm}
\end{figure}

In Fig.~\ref{Simulation_trajectory}(a), we present the trajectory of a UAV after performing $30$ cycles, whose sensing target and mobile device are located at $\textbf{\emph{x}}^t = (50,-200,0)$ and $\textbf{\emph{x}}^d = (-350,-150,0)$, respectively. We can observe that the UAV first flies towards to its sensing target, and then flies among multiple adjacent points which are close to the optimal point obtained by simulation. Besides, the UAV's trajectory is within the triangular area whose vertexes are the BS, the mobile device, and the sensing target. This is consistent with Propositions~\ref{Remark3} and~\ref{Remark4}. Moreover, the UAV's altitude is changing during the flight, which is caused by the trade-off between the path loss and the LoS probability. In Fig.~\ref{Simulation_trajectory}(b), we show the UAV's STPs at different cycles on its trajectory. Overall, the UAV moves far from its mobile device and close to the BS. Therefore, as the UAV performs cycles, the STP in the U2D mode approximately decreases while that in the cellular mode monotonously increases, and the switching point is at the $9$-th cycle. We can conclude that the UAV tends to select the U2D mode before the $9$-th cycle, after which it prefers the cellular mode, which is consistent with Propositions~\ref{Remark1} and~\ref{Remark2}. Besides, by selecting the better transmission mode, the UAV  can effectively improve its overall STP during the flight, which is illustrated by the green curve in the figure.

\begin{figure}[!t]
\centering
\vspace{-6mm}
\includegraphics[width=6in]{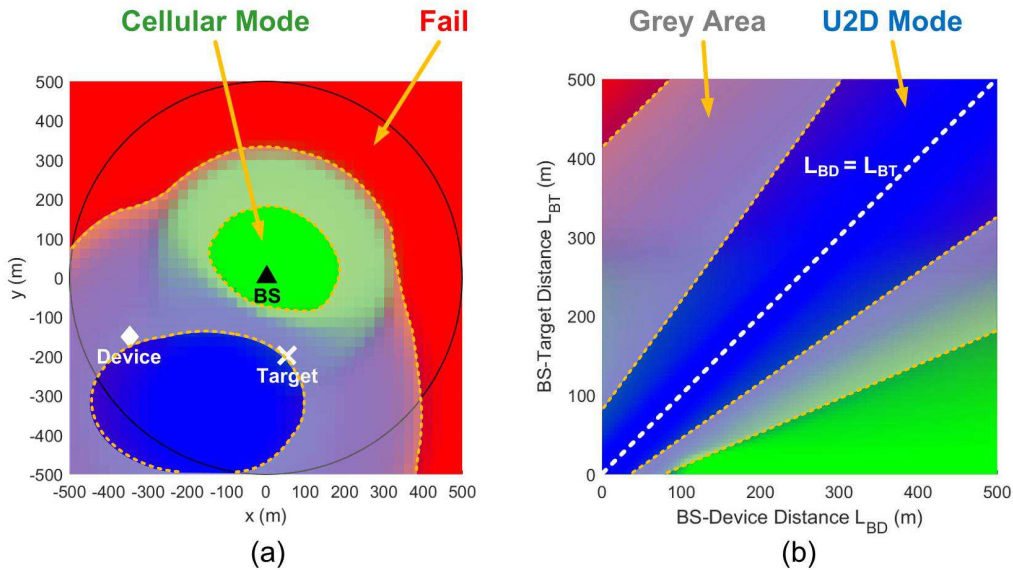}
\vspace{-12mm}
\caption{(a)~Transmission state distribution of the UAV, whose sensing target and mobile device are located at $\textbf{\emph{x}}^t = (50,-200,0)$ and $\textbf{\emph{x}}^d = (-350,-150,0)$, respectively; (b)~Transmission state distribution given different BS-device distances $L_{BD}$ and BS-target distances $L_{BT}$.}
\label{Mode_selection}
\vspace{-7mm}
\end{figure}

Fig.~\ref{Mode_selection}(a) shows the transmission state distribution of a UAV, with sensing target located at $\textbf{\emph{x}}^t = (50,-200,0)$ and mobile device located at $\textbf{\emph{x}}^d = (-350,-150,0)$. Here, we assume that other UAVs locate on their mobile devices with the altitude of $100~\textrm{m}$. From the figure, we can observe that the UAV tends to select the U2D mode for data transmission when it is close to its mobile device. Besides, as the UAV moves away from its mobile device and close to the BS, it prefers to transmit in the cellular mode. Neither of the modes can support valid data transmissions when the UAV flies far from its mobile device as well as the BS. When the UAV flies into the grey area, it may have multiple possible transmission states, i.e., the data transmission may succeed in either of the modes, or fail. In Fig.~\ref{Mode_selection}(b), we present the transmission state distribution of a UAV hovering over its sensing target at the altitude of $100~\textrm{m}$. We assume that the mobile device, the sensing target, and the BS are collinear. The BS-device distance and BS-target distance are denoted as $L_{BD}$ and $L_{BT}$, respectively. As shown in the figure, the UAV tends to select the U2D mode when its mobile device and sensing target are close to each other, i.e., the area near the white dash line which implies $L_{BD} = L_{BT}$. In addition, the UAV will transmit the sensory data in the cellular mode when the sensing target is close to the BS while the mobile device is far from the BS. The data transmission may fail when the mobile device is next to the BS while the sensing target is remote from the BS. In the grey areas, the UAV may have more than one possible transmission states.

\section{Conclusion}
%%%%%%%%%%%%%%%%%%%%%%%%%%%%%%%%%%%
\label{Conclusion}
%%%%%%%%%%%%%%%%%%%%%%%%%%%%%%%%%%%
In this paper, we have proposed overlaying U2D communications in a cellular Internet of UAVs, and investigated the trajectory design problem in this network. We have designed a joint sensing and transmission protocol to schedule the sensing and transmission processes of multiple UAVs, and analyzed this protocol using nested bi-level Markov chains. Since the trajectory design problem can be regarded as a MDP, we have formulated this problem using multi-agent DRL. Then, we have proposed a DQN-based multi-UAV trajectory design algorithm to solve this problem. Simulation results have shown that our proposed algorithm can achieve a higher utility than policy gradient algorithm and single-agent algorithm.

Three conclusions can be drawn from the simulation results. First, UAVs tend to select the U2D mode when their mobile devices and sensing targets are close to each other. Second, UAVs prefer the cellular mode if the sensing targets are near the BS while the mobile devices are located at the cell edge. Third, UAVs' transmissions are prone to fail when their sensing targets are located at the cell edge while their mobile devices are close to the central BS.

\begin{appendices}
\section{Proof of Proposition 1}\label{Proof1}
According to the joint sensing and transmission protocol in Section~\ref{Sense-and-Send Protocol}, a UAV will select the U2D mode in either of the following two cases:
\begin{itemize}
  \item The data transmissions satisfy the QoS requirement only in the U2D mode;
  \item The data transmissions satisfy the QoS requirement in both the U2D and the cellular modes, and the throughput in the U2D mode is higher than that in the cellular mode.
\end{itemize}
We can then express the probability that the $i$-th UAV selects the U2D mode by $\mathbb{P}^u_{ts,i} = \mathcal{P}^{u}_{ts,i}\left(1-\mathcal{P}^{c}_{ts,i}\right) + \mathcal{P}^{u}_{ts,i} \mathcal{P}^{c}_{ts,i} \cdot \mathcal{P}\left\{R^u_{i} \geq R^c_{i} \right\}$. Similarly, the $i$-th UAV selects the cellular mode with probability $\mathbb{P}^c_{ts,i} = \left(1-\mathcal{P}^{u}_{ts,i}\right)\mathcal{P}^{c}_{ts,i} + \mathcal{P}^{u}_{ts,i} \mathcal{P}^{c}_{ts,i} \cdot \mathcal{P}\left\{R^u_{i} < R^c_{i} \right\}$. Given the location of the $i$-th UAV $\textbf{\emph{x}}_i$, the location of its mobile device $\textbf{\emph{x}}^d_i$ and the location of the BS $\textbf{\emph{x}}_0$, the terms $\mathcal{P}\left\{R^u_{i} < R^c_{i} \right\}$ and $\mathcal{P}\left\{R^u_{i} \geq R^c_{i} \right\}$ can be expressed as
\begin{align}\label{R_U_R_C}
\mathcal{P}\left\{ R^u_{i} < R^c_{i} \right\} &= \mathcal{P}\left\{ R^u_{i} < \frac{1}{2} \textmd{min}\left\{ R^c_{a,i}, R^c_{g,i} \right\} \right\} = \mathcal{P}\left\{ 2R^u_{i} < R^c_{a,i} \right\} \mathcal{P}\left\{ 2R^u_{i} < R^c_{g,i} \right\} \nonumber \\
&= \mathbb{E} \left\{ \mathcal{P}_{a,i}(\textbf{\emph{x}}_0, \textbf{\emph{x}}_i, 2R^u_{i}) \mathcal{P}_{g,i}(\textbf{\emph{x}}^d_i,\textbf{\emph{x}}_0,2R^u_{i}) \right\} \mathop{\rm{=}}\limits^\triangle \mathbb{E} \left\{ \Psi_i(R^u_{i}) \right\},
\end{align}
and $\mathcal{P}\left\{R^u_{i} \geq R^c_{i} \right\} = 1 - \mathcal{P}\left\{R^u_{i} < R^c_{i} \right\}$. Here, $\mathbb{E} \left\{ \cdot \right\}$ is the mathematical expectation.

Based on the assumptions in Section~\ref{System Model}, if the LoS component exists in the link from $\textbf{\emph{x}}_i$ to $\textbf{\emph{x}}^d_i$, $\zeta_{LoS,i}$ follows the Rice distribution with probability density function~(PDF) $f _{\zeta_{LoS,i}} \left( {x } \right) = \frac{x}{{{\sigma ^2}}}\exp \left( { - \frac{{{x^2} + {A^2}}}{{2{\sigma ^2}}}} \right) \cdot {I_0}\left( {\frac{xA}{\sigma ^2}} \right)$. Here, $I_0(\cdot)$ denote the modified Bessel function of the first kind with order zero. Besides, $A$ and $\sigma$ are parameters, which can be obtained through $\Omega = A^2 + 2\sigma^2 = 1$ and $K_{ri} = \frac{A^2}{2\sigma ^2}$. Thus, the PDF of $R^u_{i}$ is $f_{R_{LoS,i}^u}\left( y \right) = f_{\zeta_{LoS,i}} \left( \frac{2^y - 1}{\xi _{LoS,i}} \right) \frac{2^y\ln 2}{\xi _{LoS,i}}$, where $\xi _{LoS,i} = \frac{P^u}{N_0 \cdot 10^{\mathcal {L}_{LoS,i}/10}}$. Likewise, when the NLoS component exists, $\zeta_{NLoS,i}$ is Rayleigh-distributed with PDF $f _{\zeta_{NLoS,i}} \left( {x } \right) = \frac{x}{\sigma ^2}\exp \left( { - \frac{x^2} {2\sigma ^2}} \right)$, where $\sigma$ is a parameter with the value of one. Then, the PDF of $R^u_{i}$ is given by $f_{R_{NLoS,i}^u}\left( y \right) = f_{\zeta_{NLoS,i}} \left( \frac{2^y - 1}{\xi _{NLoS,i}} \right) \frac{2^y\ln 2}{\xi _{NLoS,i}}$, where $\xi _{NLoS,i} = \frac{P^u}{N_0 \cdot 10^{\mathcal {L}_{NLoS,i}/10}}$. Therefore, given probabilities $\mathcal{P}_{LoS,i}$ and $\mathcal{P}_{NLoS,i}$, we can express the PDF of $R^u_{i}$ as $f_{R_i^u}\left( y \right) = \mathcal{P}_{LoS,i} \cdot f_{R_{LoS,i}^u}\left( y \right) + \mathcal{P}_{NLoS,i} \cdot f_{R_{NLoS,i}^u}\left( y \right)$. Finally, we have
$\mathcal{P}\left\{R^u_{i} < R^c_{i} \right\} =  \int_{-\infty}^\infty \Psi_i(y) f_{R_i^u}\left( y \right)dy$
and $\mathcal{P}\left\{R^u_{i} \geq R^c_{i} \right\} = 1-\int_{-\infty}^\infty \Psi_i(y) f_{R_i^u}\left( y \right)dy$, where $\Psi_i(y)$ is defined in equation~(\ref{R_U_R_C}). Then, $\mathbb{P}^u_{ts,i}$ and $\mathbb{P}^c_{ts,i}$ can be obtained, and the proof ends.

\end{appendices}

%=============================================bibliography======================================================


\begin{thebibliography}{40}
%%%%%%%%%%%%%%%%%%%%%%%%%%%%%
% Introduction(background)
%%%%%%%%%%%%%%%%%%%%%%%%%%%%%
\bibitem{FHJL-2019}
F.~Wu, H.~Zhang, J.~Wu, and L.~Song, ``Trajectory design for overlay UAV-to-Device communications by deep reinforcement learning,'' in \emph{Proc. IEEE GLOBECOM}, Waikoloa, HI, USA, Dec.~2019.

\bibitem{HJJYG-2018}
H.~Wang, J.~Wang, J.~Chen, Y.~Gong, and G.~Ding, ``Network-connected UAV communications: Potentials and challenges'', \emph{China Commun.}, vol.~15, no.~12, pp.~111-121, Dec.~2018.

\bibitem{JCZYRL-2017}
J.~Wang, C.~Jiang, Z.~Han, Y.~Ren, R.~G.~Maunder, and L.~Hanzo, ``Taking drones to the next level: Cooperative distributed unmanned-aerial-vehicular networks for small and mini drones'', \emph{IEEE Veh. Technol. Mag.}, vol.~12, no.~3, pp.~73-82, Sep.~2017.

\bibitem{KGMK-2013}
K.~Kanistras, G.~Martins, M.~J.~Rutherford, and K.~P.~Valavanis, ``A survey of unmanned aerial vehicles~(UAVs) for traffic monitoring,'' in \emph{Proc. IEEE ICUAS}, Atlanta, GA, USA, May~2013.

\bibitem{JMJSCR-2013}
J.~Nikolic, M.~Burri, J.~Rehder, S.~Leutenegger, C.~Huerzeler, and R.~Siegwart, ``A UAV system for inspection of industrial facilities,'' in \emph{Proc. IEEE Aerosp. Conf.}, Big Sky, MT, USA, Mar.~2013.

\bibitem{BKDF-2017}
B.~H.~Y.~Alsalam, K.~Morton, D.~Campbell, and F.~Gonzalez, ``Autonomous UAV with vision based on-board decision making for remote sensing and precision agriculture,'' in \emph{Proc. IEEE Aerosp. Conf.}, Big Sky, MT, USA, Mar.~2017.

\bibitem{CZY-2017}
C.~Yuan, Z.~Liu, and Y.~Zhang, ``Fire detection using infrared images for UAV-based forest fire surveillance,'' in \emph{Proc. IEEE ICUAS}, Miami, FL, USA, Jun.~2017.

\bibitem{RJW-2013}
R.~J.~Kerczewski, J.~D.~Wilson, and W.~D.~Bishop, ``Frequency spectrum for integration of unmanned aircraft,'' in \emph{Proc. IEEE/AIAA DASC}, East Syracuse, NY, USA, Oct.~2013.

\bibitem{3GPP_TR_36_777}
3GPP TR 36.777, ``Enhanced LTE support for aerial vehicles,'' Release 15, Dec.~2017.

\bibitem{HLZH-2019}
H.~Zhang, L.~Song, Z.~Han, and H.~V.~Poor, ``Cooperation techniques for a cellular internet of unmanned aerial vehicles,'' \emph{IEEE Wireless Commun.},  vol.~26, no.~5, pp.~167-173, Oct.~2019.

\begin{comment}
\bibitem{HYL-2017}
H.~Zhang, Y.~Liao and L.~Song, ``D2D-U: Device-to-device communications in unlicensed bands for 5G system", \emph{IEEE Trans. Wireless Commun.}, vol.~16, no.~6, pp.~3507-3519, Jun.~2017.
\end{comment}

%%%%%%%%%%%%%%%%%%%%%%%%%%%%%
% Introduction(system)
%%%%%%%%%%%%%%%%%%%%%%%%%%%%%

\bibitem{PX-2017}
P.~Jia, and X.~Wang, ``Nested markov chain - A novel approach to model network-induced constraints,'' in \emph{Proc. IEEE IEMCON}, Vancouver, BC, Canada, Oct.~2017.

\bibitem{JHLZH-2019}
J.~Hu, H.~Zhang, L.~Song, Z.~Han, and H.~V.~Poor, ``Reinforcement learning for a cellular Internet of UAVs: Protocol design, trajectory control, and resource management,'' 2019, [online] Available: https://arxiv.org/abs/1911.08771.

\bibitem{RA-1998}
R.~S.~Sutton, and A.~G.~Barto, ``Reinforcement learning: An introduction.'' MIT press, Cambridge, MA, Sep.~1998.

\bibitem{KMMA-2017}
K.~Arulkumaran, M.~P.~Deisenroth, M.~Brundage, and A.~A.~Bharath, ``Deep reinforcement learning: A brief survey'' \emph{IEEE Signal Processing Mag.}, vol.~34, no.~6, pp.~26-38, Nov.~2017.

\bibitem{Y-2017}
%Y. Li, ``\emph{Deep reinforcement learning: An overview}'', 2017, [online] Available: https://arxiv.org/abs/1701.07274.
Z.~Xiong, Y.~Zhang, D.~Niyato, R.~Deng, P.~Wang, and L.~Wang, "Deep reinforcement learning for mobile 5G and beyond: Fundamentals, applications and challenges", \emph{IEEE Veh. Technol. Mag.}, vol.~14, no.~2, pp. 44-52, 2019.

%%%%%%%%%%%%%%%%%%%%%%%%%%%%%
% Introduction(review)
%%%%%%%%%%%%%%%%%%%%%%%%%%%%%

\bibitem{SHBL2-2018}%U2N
S.~Zhang, H.~Zhang, B.~Di, and L.~Song, ``Joint trajectory and power optimization for UAV sensing over cellular networks,'' \emph{IEEE Commun. Lett.}, vol.~22, no.~11, pp.~2382-2385, Nov.~2018.

\begin{comment}
\bibitem{AK-2018}%U2N
A.~Al-Hourani, and K.~Gomez, ``Modeling cellular-to-UAV path-loss for suburban environments,'' \emph{IEEE Wireless Commun. Lett.}, vol.~7, no.~1, pp.~82-85, Feb.~2018.

\bibitem{ZJLKG-2018}%U2U
Z.~Yuan, J.~Jin, L.~Sun, K.~Chin, and G.~Muntean, ``Ultra-reliable IoT communications with UAVs: A swarm use case,'' \emph{IEEE Commun. Mag.}, vol.~56, no.~12, pp.~90-96, Dec.~2018.
\end{comment}

\bibitem{SHBL3-2018}%U2U
S.~Zhang, H.~Zhang, B.~Di, and L.~Song, ``Cellular UAV-to-X communications: Design and optimization for multi-UAV networks,'' \emph{IEEE Trans. Wireless Commun.}, vol.~18, no.~2, pp.~1346-1359, Feb.~2019.

\bibitem{SHQKL-2018}
S.~Zhang, H.~Zhang, Q.~He, K.~Bian, and L.~Song, ``Joint trajectory and power optimization for UAV relay networks,'' \emph{IEEE Commun. Lett.}, vol.~22, no.~1, pp.~161-164, Jan.~2018.

\bibitem{MMWCMC-2017}
M.~Chen, M.~Mozaffari, W.~Saad, C.~Yin, M.~Debbah, and C.~S.~Hong, ``Caching in the sky: Proactive deployment of cache-enabled unmanned aerial vehicles for optimized quality-of-experience'', \emph{IEEE J.~Sel.~Areas~Commun.},  vol.~35, no.~5, pp.~1046-1061, May~2017.

\bibitem{ZYNDPN-2020}
Z.~Xiong, Y.~Zhang, N.~C.~Luong, D.~Niyato, P.~Wang, and N.~Guizani, ``The best of both worlds: A general architecture for data management in blockchain-enabled Internet-of-Things," \emph{IEEE Network}, vol.~34, no.~1, pp.~166-173, Jan.~2020.

\bibitem{ISH-2018}
I.~Bor-Yaliniz, S.~S.~Szyszkowicz, and H.~Yanikomeroglu, ``Strategic densification with UAV-BSs in cellular networks,'' \emph{IEEE Wireless Commun. Lett.}, vol.~7, no.~3, pp.~372-375, Jun.~2018.

\bibitem{FIH-2018}
F.~Lagum, I.~Bor-Yaliniz, and H.~Yanikomeroglu, ``Environment-aware drone-base-station placements in modern metropolitans,'' \emph{IEEE Wireless Commun. Lett.}, vol.~7, no.~3, pp.~384-387, Jun.~2018.

\bibitem{IAH-2019}
I.~Bor-Yaliniz, A.~El-Keyi, H.~Yanikomeroglu, ``Spatial configuration of agile wireless networks with drone BSs and user-in-the-loop'', \emph{IEEE Trans. Wireless Commun.}, vol.~18, no.~2, pp.~753-768, Feb.~2019.

\bibitem{UWC-2019}
U.~Challita, W.~Saad, and C.~Bettstetter, ``Interference management for cellular-connected UAVs: A deep reinforcement learning approach,'' \emph{IEEE Trans. Wireless Commun.}, vol.~18, no.~4, pp.~2125-2140, Apr.~2019.

\bibitem{MWC-2019}
M.~Chen, W.~Saad, and C.~Yin, ``Liquid state machine learning for resource and cache management in LTE-U unmanned aerial vehicle (UAV) networks,'' \emph{IEEE Trans. Wireless Commun.}, vol.~18, no.~3, pp. 1504-1517, Mar.~2019.

%%%%%%%%%%%%%%%%%%%%
% Model & Protocol
%%%%%%%%%%%%%%%%%%%%

\bibitem{VI-2017}
V.~V.~Shakhov, and I.~Koo, ``Experiment design for parameter estimation in probabilistic sensing models,'' \emph{IEEE Sensors J.}, vol.~17, no.~24, pp.8431-8437, Oct.~2017.

\bibitem{ARAS-2013}
A.~Chakraborty, R.~R.~Rout, A.~Chakrabarti, and S.~K.~Ghosh, ``On network lifetime expectancy with realistic sensing and traffic generation model in wireless sensor networks,'' \emph{IEEE Sensors J.}, vol.~13, no.~7, pp.~2771-2779, Apr.~2013.

\bibitem{SJHL-2019}
S.~Zhang, J.~Yang, H.~Zhang, and L.~Song, ``Dual trajectory optimization for a cooperative internet of UAVs,'' \emph{IEEE Commun. Lett.}, vol.~23, no.~6, pp.~1093-1096, Jun.~2019.

\bibitem{HLZ-2020}
H.~Zhang, L.~Song, and Z.~Han, \emph{Unmanned aerial vehicle applications over cellular networks for 5G and beyond}, New York, NY, USA: Springer, 2020.

\bibitem{S-1944}
S.~O.~Rice, ``Mathematical analysis of random noise,'' \emph{Bell Syst. Tech. J.}, vol.~23, no.~3, pp.~282-332, Jul.~1944.

\bibitem{J-1950}
J.~I.~Marcum, ``Table of Q functions,'' Rand Corp., Santa Monica, CA, Tech. Rep. U.S. Air Force Project RAND Res. Memo. M-339, ASTIA Document AD 1165451, Jan.~1950.

\bibitem{WCPP}
T.~S.~Rappaport, \emph{Wireless communications: Principles and practice}. Englewood Cliffs, NJ, USA: Prentice-Hall, 1996.

\bibitem{JHL-2018}
J.~Hu, H.~Zhang, and L.~Song, ``Reinforcement learning for decentralized trajectory design in cellular UAV networks with sense-and-send protocol,'' \emph{IEEE Internet Things J.}, vol.~6, no.~4, pp.~6177-6189, Aug.~2019.

%%%%%%%%%%%%%%%%%%%%
% Algorithm
%%%%%%%%%%%%%%%%%%%%

\bibitem{CP-1992}
C.~Watkins, and P.~Dayan, ``Q-learning,''\emph{ Mach. Learn.}, vol.~8, pp.~279-292, 1992.

\bibitem{J-1994}
J.~Tsitsiklis, ``Asynchronous stochastic approximation and Q-learning,'' \emph{Mach. Learn.}, vol.~16, pp.~185-202, 1994.

\bibitem{HG-2018}
H.~Ye, and G.~Y.~Li, ``Deep reinforcement learning based distributed resource allocation for V2V broadcasting,'' in \emph{Proc. IEEE IWCMC}, Limassol, Cyprus, Jun.~2018.

\bibitem{nature-2015}
V.~Mnih, K.~Kavukcuoglu, D.~Silver, A.~A.~Rusu, J.~Veness, M.~G.~Bellemare, A.~Graves, M.~Riedmiller, A.~K.~Fidjeland, G.~Ostrovski, S.~Petersen, C.~Beattie, A.~Sadik, I.~Antonoglou, H.~King, D.~Kumaran, D.~Wierstra, S.~Legg, D.~H.~I.~Antonoglou, D.~Wierstra, and M.~A.~Riedmiller, ``Human-level control through deep reinforcement learning,'' \emph{Nature}, vol.~518, no.~7540, pp.~529-533, Feb.~2015.

\bibitem{RTZ-2013}
R.~Zheng, T.~Le, and Z.~Han, ``Approximate online learning for passive monitoring of multi-channel wireless networks,'' in \emph{Proc. IEEE INFOCOM}, Turin, Italy, Apr.~2013.

\bibitem{RDSY-1999}
R.~S.~Sutton, D.~A.~McAllester, S.~P.~Singh, and Y.~Mansour, ``Policy gradient methods for reinforcement learning with function approximation,'' \emph{Advances in neural information processing systems}, pp.~1057-1063, 1999.



%%%%%%%%%%%%%%%%%%%%
% Simulation
%%%%%%%%%%%%%%%%%%%%

\bibitem{GTG-2013}
G.~E.~Dahl, T.~N.~Sainath, and G.~E.~Hinton, ``Improving deep neural networks for LVCSR using rectified linear units and dropout,'' in \emph{Proc. IEEE ICASSP}, Vancouver, BC, Canada, May~2013.



\end{thebibliography}
\end{document}